\documentclass{article}

\usepackage[left=0.9in, bottom=0.9in, right=0.9in, top=0.9in]{geometry}
\usepackage{amsmath}
\usepackage{hyperref}
\usepackage{etoolbox}
\usepackage{bbm}
\usepackage{amssymb, eufrak, euscript}
\usepackage{booktabs}
\usepackage{amsmath}
\usepackage{graphicx,psfrag,epsf}
\usepackage{enumerate}
\usepackage[numbers, square]{natbib}
\usepackage{dsfont}
\usepackage{caption}
\usepackage{epstopdf}
\usepackage{xcolor, soul}
\usepackage{subcaption}
\usepackage{lscape}
\usepackage{setspace} 
\usepackage{longtable}
\usepackage[colorinlistoftodos,prependcaption,textsize=tiny]{todonotes}
\usepackage{url} 
\usepackage[T1]{fontenc}
\usepackage[utf8]{inputenc}
\usepackage{authblk}

\newcommand{\argmin}{\operatornamewithlimits{arg\,min}}
\newtheorem{theorem}{Theorem}
\newenvironment{proof}[1][Proof]{\textbf{#1.} }{\ \rule{0.5em}{0.5em}}

\title{\bf Exploring Financial Networks Using Quantile Regression and Granger Causality}

\author[1]{Kara Karpman \footnote{Equal contribution}
}
\author[1]{Samriddha Lahiry $^*$}
\author[2]{Diganta Mukherjee}
\author[1]{Sumanta Basu  \footnote{Corresponding author. Email: \href{mailto:sumbose@cornell.edu}{sumbose@cornell.edu}}}

\affil[1]{Department of Statistics and Data Science, Cornell University}
\affil[2]{Sampling and Official Statistics Unit, Indian Statistical Institute Kolkata}

\date{ }

\begin{document}
\maketitle

\begin{abstract}
In the post-crisis era, financial regulators and policymakers are increasingly interested in data-driven tools to measure systemic risk and to identify systemically important firms. Granger Causality (GC) based techniques to build networks among financial firms using time series of their stock returns have received significant attention in recent years. Existing GC network methods model conditional means, and do not distinguish between connectivity in lower and upper tails of the return distribution - an aspect crucial for systemic risk analysis. We propose statistical methods that measure connectivity in the financial sector using system-wide tail-based analysis and is able to distinguish between  connectivity in lower and upper tails of the return distribution. This is achieved using bivariate and multivariate GC analysis based on regular and Lasso penalized quantile regressions, an approach we call quantile Granger causality (QGC).  By considering centrality measures of these financial networks, we can assess the build-up of systemic risk and identify risk propagation channels. We provide an asymptotic theory of QGC estimators under a quantile vector autoregressive model, and show its benefit over regular GC analysis on simulated data.  We apply our method to the monthly stock returns of large U.S. firms and demonstrate that lower tail based networks can detect systemically risky periods in historical data with higher accuracy than mean-based networks. In a similar analysis of large Indian banks, we find that upper and lower tail networks convey different information and have the potential to distinguish between periods of high connectivity that are governed by positive vs negative news in the market.  

\end{abstract}

\section{Introduction}\label{sec:intro}

Understanding complex linkages among market participants in an interlinked financial market is of interest to  researchers and policy makers in financial economics. Since it is often difficult to access data on firms' balance sheet and counterparty transactions in real time, there is considerable interest in learning the structure of financial networks in a data-driven fashion \citep{billio2012econometric, diebold2014network}. Data-driven financial networks have been empirically successful for two types of analyses. First, these networks tend to be denser during the periods of financial crisis, providing a way to monitor systemic risk in the market. Second, firms with high network centrality in and around crisis period are deemed to be systemically important. For these reasons, it is important to develop statistical methods capable of discovering nuanced connectivity structure among financial firms from data.

In a typical data-driven financial network, each node represents a firm, and an edge between two nodes encode some form of ``relationship'' between the historical time series of the two  firms' health characteristics (e.g. stock returns or realized volatility). Broadly speaking, two types of relationships are commonly explored in the literature - contemporaneous association such as correlation or co-movement \citep{adrian2011covar, hardle2016tenet}, and lead-lag or Granger causality (GC) patterns showing one firms' data can be used to predict the behavior of the other firm \citep{billio2012econometric, diebold2014network}. In this paper, we focus on the second type of relationships.

In the context of modeling stock returns, a key object of interest is the tail risk captured by quantiles of return distributions. While a number of works have built financial networks based on contemporaneous association among the tail risks of firms, the existing literature on GC-based networks has predominantly focused on mean returns instead of predictability in the tails. To narrow this gap, in this paper we propose quantile Granger causality (QGC), which combines quantile regression and GC to build quantile-specific financial networks. We propose a pairwise and a system-wide variant of QGC. The former uses bivariate analyses on two firms' returns at a time, while the latter jointly analyzes all firms' returns with Lasso penalization to account for spurious connectivity patterns.
\begin{figure}[!h]
  \centering
  \includegraphics[width = 0.62\textwidth, trim = {0 3in 0 1.7in}, clip]{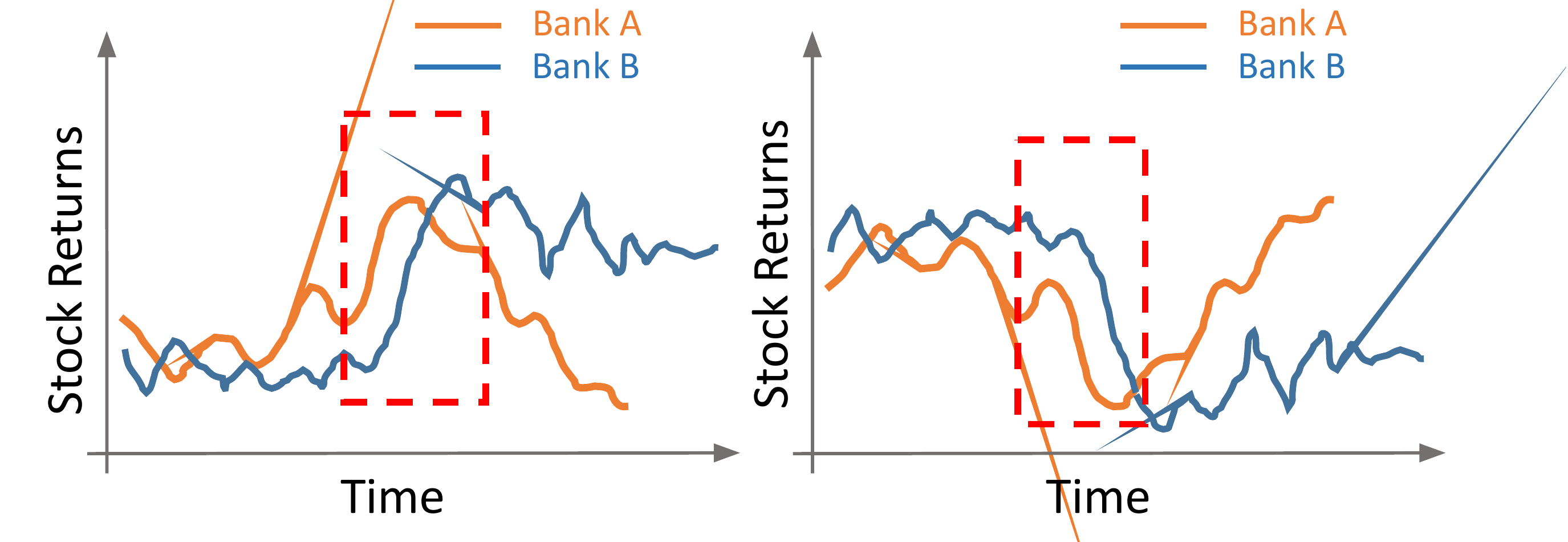}
  \includegraphics[width = 0.36\textwidth]{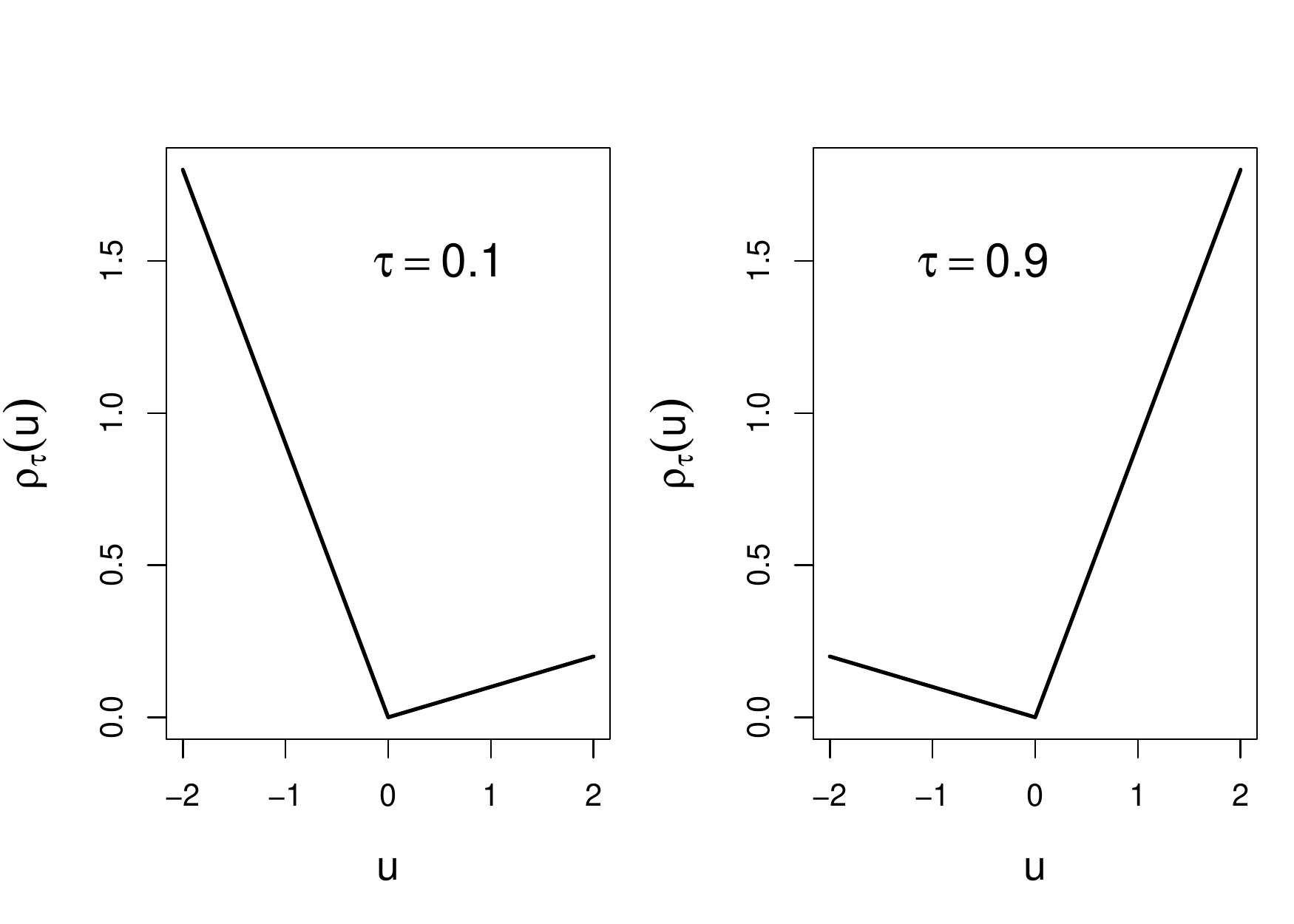}
  \caption{[Left]: GC relationship from bank A to bank B exists only at a higher or a lower tail of return distribution. [Right]: asymmetric check loss function designed to capture connections that are prominent at different quantiles $\tau$.}
  \label{fig:QGC_diagram}
\end{figure}
Our main premise is that by building on financial networks separately for upper and lower tails of return distributions, QGC can capture nuanced linkages among firms that are not prominent in mean-based GC networks. For instance, linkages amongst firms that exist only in bad days of the market will be accentuated by QGC, and we can identify central firms which play important roles primarily in crisis periods (see the left two plots in Figure \ref{fig:QGC_diagram}). Tail-specific networks could also provide new insight in risk monitoring. The existing literature often associate periods of high connectivity with periods of financial stress. However, focusing separately on upper and lower tail QGC networks can help periods of high connectivity which reflect shared confidence of market participants in the economy.

To build tail-specific financial networks among $p$ firms from a $p$-dimensional time series of their stock returns $\{\mathbf{x}_t\}_{t=1}^n$, we model conditional $\tau^{th}$ quantile of stock returns $Q_\tau(x_{i,t+1}|\mathbf{x}_t)$ for firm $i$ as a linear function of $\mathbf{x}_t$. A non-zero regression coefficient on $\mathbf{x}_{j,t}$ means there is an edge from $j$ to $i$ in the network. This is in contrast with regular GC-based networks, which model the conditional mean $\mathbb{E}(x_{i,t+1}|\mathbf{x}_t)$ instead. The conditional quantiles are modelled by using \textit{quantile regression} \citep{quantileregressionKoenker} that changes the symmetric squared error loss of regular regression to an asymmetric check loss function $\rho_\tau(u) = u(\tau - \mathds{1}[u \le 0])$. For smaller values of $\tau$, this loss function upweights negative losses and is minimized by a lower quantile of the return distribution (See the right two plots in Figure \ref{fig:QGC_diagram}).

We also provide a systematic asympotic analysis of Lasso penalized QGC estimators under a quantile vector autoregressive model \citep{koenker_xiao} in a fixed $p$ asymptotic regime. To the best of our knowledge, these are new results. Then we illustrate the finite-sample performance of QGC methods on simulated data. We find that multivariate QGC networks offer higher accuracy in detecting central firms in a hub-structured financial network.

Finally, we use both bivariate and Lasso penalized multivariate quantile regression to estimate financial networks for upper and lower quantiles. In particular, we explore the  evolution of network connectivity in two data sets on stock returns: (i) 75 large US financial firms (from banks, broker-dealer and insurance sectors), and (ii) 30 Indian banks, over 36-month rolling windows spanning nearly two decades. 
We find that QGC offers interesting additional insights into the structure of linkages that are not offered by mean-based GC networks. In the analysis of Indian banks, QGC was able to distinguish between periods of high connectivity aligned with negative news in the market from the periods of high connectivity aligned with positive news (see Figure \ref{fig:india_bank_time_series_avg_degree}). This is in sharp contrast with the current interpretation of mean-based GC networks that  always associates high connectivity with periods of systemic stress. In the analysis of US financial firms,  lower-tail QGC networks were able to detect periods of systemic stress with higher sensitivity than mean-based GC networks. 

The rest of the paper is organized as follows. Section \ref{sec:method} provides a description of bivariate and lasso penalized QGC estimation methods, and Section \ref{sec:theory} provides some asymptotic analysis. Performance of QGC on simulated data is investigated in Section \ref{sec:simulation}, and Sections \ref{sec:US_empirical_section} and \ref{sec:India_empirical_section} contain the empirical analyses on US and Indian firms.

\section{Methods}\label{sec:method}
We start with a brief review of bivariate and multivariate GC based methods in the literature for building financial networks, and then describe two new methods based on quantile regressions.

\subsection{Background: Bivariate and Multivariate Granger Causality Networks}
Bivariate GC methods for financial networks were proposed originally by \citet{billio2012econometric}, who construct networks of financial firms using bivariate GC  \citep{Granger}.  In this framework, two firms are connected if the stock returns of one have predictive power over the stock returns of the other.

Let $\{x_{i,t}\}_{t=1}^{n+1}$ and $\{x_{j,t}\}_{t=1}^{n+1}$ denote (stationary) time series of the stock returns of firms $i$ and $j$.  Consider a model in which each of these firms' returns is centered around a linear combination of lagged returns; that is,
\begin{align}\label{eq:GC 1}
x_{i,t+1} &= \alpha_{ii}x_{i,t} + \alpha_{ij}x_{j,t} + \epsilon_{i,t+1}, \\ \label{eq:GC2}
x_{j,t+1} &= \alpha_{ji}x_{i,t} + \alpha_{jj}x_{j,t}+ \epsilon_{j,t+1},
\end{align}    
with $\epsilon_{k,t+1} \overset{i.i.d.}\sim (0,\sigma_k^2)$ for $k \in \{i,j\}$.  Then firm $j$'s returns are said to Granger-cause firm $i$'s returns if $\alpha_{ij} \neq 0$, meaning firm $j$'s lagged returns help predict firm $i$'s returns over and above firm $i$'s own lagged returns.  An analogous statement can be made for firm $i$'s returns Granger-causing firm $j$'s returns.  

In \citet{billio2012econometric}, the authors form bivariate linear models in the form of (\ref{eq:GC 1})-(\ref{eq:GC2}) for all pairs of firms, $(i,j)$, in their sample.  They then construct networks whose nodes are firms and where an edge exists between nodes $i$ and $j$ if and only if
\begin{align*}
\max\left\{|{\alpha}_{ij}|, |{\alpha}_{ji}|\right\} \neq 0.
\end{align*}
Empirically the authors obtained linear regression estimates, $\hat{\alpha}_{ij}$ and $\hat{\alpha}_{ji}$, and tested if either estimate was non-zero at a 5\% significance level. The end result is an undirected network among firms with edges based on the above-described lead-lag relationships.  

Bivariate GC networks are constructed adopting a \textit{pairwise} approach, i.e. data from only two firms $i$ and $j$ are used without considering potential effects from a third firm $k$. This can lead to false positive network edges. A number of works in this field has adopted a more \textit{system-wide} view to tackle this issue. Here the stock returns (or volatilities) of all $p$ firms are modeled jointly as a multivariate time series, and an edge exists between two firms $i$ and $j$ if firm $j$ can predict the future of firm $i$ even after accounting for the present returns of all the other firms in the system. A vector autoregressive (VAR) model is used to estimate such GC networks. Formally, the bivariate model is  generalized to a multivariate model 
\begin{align}\label{eq: multivar_mean}
\mathbb{E}\left(x_{i,t+1} \big\vert \{x_{1,t},....,x_{p,t}\}\right) = 
\sum_{\ell=1}^p \alpha_{i\ell}x_{\ell,t} = \mathbf{x}_{t}'\boldsymbol\alpha_i,
\end{align}
and an edge exists between nodes $i$ and $j$ if $\max\{|\alpha_{ij}|, |\alpha_{ji}| \} \neq 0$. Note that a non-zero coefficient $\alpha_{ij}$ signifies that the returns of firm $j$ can predict the return of firm $i$ \textit{even after accounting for the returns of all the other firms in the system}.

To address the issue of high-dimensionality arising from including all $p$ firms in the model, different penalized regressions of sparsity inducing priors have been used in the literature  \citep{demirer2018estimating, preprint, ahelegbey2016sparse}. 
For example, in a Lasso penalized multivariate GC network, the edges to node $i$ are obtained by solving the following optimization problem
\begin{align}\label{eq:optimization_mean}
\min_{\boldsymbol\alpha \in \mathbb{R}^p}\left[\frac{1}{n}\sum_{t=1}^n\left(x_{i,t+1}-\mathbf{x}_{t}'\boldsymbol\alpha\right)^2 + \lambda_i \sum_{\ell = 1}^p |\alpha_{\ell}|\right],
\end{align}
where $\lambda_i$ is a tuning parameter that adjusts the amount of penalization \citep{preprint}.

\subsection{Method I: Bivariate Quantile Regression Networks}

We develop an analogous, quantile-based method that can potentially capture linkages driving financial crises, rather than those that exist at the center of the return distribution.  Notice that Granger causality represents causality in the mean: equation (\ref{eq:GC 1}) is equivalent to 
\begin{align}\label{eq: cond mean}
\mathbb{E}\left[x_{i,t+1} | x_{i,t},x_{j,t}\right] = \alpha_{ii}x_{i,t} + \alpha_{ij}x_{j,t}.
\end{align}
Thus non-zero $\alpha_{ij}$ indicates that firm $i$'s average return depends non-trivially on firm $j$'s lagged return.  However, given that we are interested in financial crises, considering a firm's average performance is insufficient.  Instead, we need to examine the returns of firms on their worst-performing days; that is, returns in the lower tail of the distribution.  Analogous to (\ref{eq: cond mean}), we express the conditional $\tau$-quantiles of $x_{i,t+1}$ and $x_{j,t+1}$ as linear combinations of lagged returns:
\begin{align}\label{eq: cond quant 1}
Q_\tau\left(x_{i,t+1} \big\vert x_{i,t},x_{j,t}\right) &= \beta_{ii}x_{i,t} + \beta_{ij}x_{j,t}, \\
\label{eq: cond quant 2}
Q_\tau\left(x_{j,t+1} \big\vert x_{i,t},x_{j,t}\right) &= \beta_{ji}x_{i,t} + \beta_{jj}x_{j,t}.
\end{align}
Then, if we take $\tau$ small (e.g. $\tau = 0.05$),  equation (\ref{eq: cond quant 1}) captures how firm $j$'s returns impact firm $i$'s returns on the latter's worst-performing days.  Continuing the analogy to Granger causality, we can form a financial network by placing an edge between $i$ and $j$ if 
\begin{align*}
\max\left\{|{\beta}_{ij}|, |{\beta}_{ji}|\right\} \neq 0.
\end{align*}

The model in (\ref{eq: cond quant 1})-(\ref{eq: cond quant 2}) can be estimated using quantile regression \citep{quantileregressionKoenker}.  Coefficient estimates are obtained by minimizing an asymmetric loss function, $\rho_\tau(u) = u\left(\tau - \mathds{1}\{u \leq 0\}\right)$, as opposed to the squared error loss function used in linear regression.  In particular, we write $\boldsymbol\beta_i = (\beta_{ii}, \beta_{ij})^\prime$ and take
\begin{align*}
\hat{\boldsymbol \beta}_i \in \argmin_{\boldsymbol\beta \in \mathbb{R}^2}\left[\frac{1}{n}\sum_{t=1}^n\rho_\tau\left(x_{i,t+1}-\mathbf{x}_{t}'\boldsymbol\beta\right)\right], 
\end{align*}
where $\mathbf{x}_{t} = (x_{i,t}, x_{j,t})'$.  The resulting network describes how each firm depends on others at the lower tail of its conditional returns distribution.  

\subsection{Method II: Multivariate Quantile Regression Networks}
\label{sec:lasso_penalized_multi_QR}
Next we refine the previously described method by enlarging the conditioning set, thereby ensuring that we only capture direct relationships between firms.  Indeed this is a major disadvantage of the bivariate approach, which may characterize two firms as being connected when they have only an indirect relationship.  For example, suppose firm $k$'s lagged returns impact firms $i$ and $j$'s current returns.  Then $\hat{\beta}_{ij}, \hat{\beta}_{ji}$ as estimated by bivariate quantile regression will be non-zero due to the fact that $i$ and $j$'s returns are being driven by a common source.  The network (which we hope illustrates only direct connections) will have a spurious link between $i$ and $j$. 

To correct for spurious connectivity, we need to condition on the lagged returns of \textit{all} firms in our sample, rather than taking the pairwise approach of equations (\ref{eq: cond quant 1})-(\ref{eq: cond quant 2}).  In other words, the model in (\ref{eq: cond quant 1}) can be generalized to 
\begin{align}\label{eq: multivar cond quant}
Q_\tau\left(x_{i,t+1} \big\vert \{x_{1,t},....,x_{p,t}\}\right) = 
\sum_{\ell=1}^p \beta_{i\ell}x_{\ell,t} = \mathbf{x}_{t}'\boldsymbol\beta_i,
\end{align}
where we have assumed that there are $p$ firms in our sample.  Let us return to the problematic example in which $i$ and $j$ were both influenced by $k$.  Under our new framework, (\ref{eq: multivar cond quant}), we will have ${\beta}_{ij} = 0 = {\beta}_{ji}$ and ${\beta}_{ik} \neq 0$; that is, the indirect connection will be eliminated and all of the weight placed on the firm whose lagged returns have a direct influence on firms $i$ and $j$'s returns.   

On the other hand, by including lagged returns of all the firms, we face the curse of dimensionality, i.e. the number of firms, $p$, may exceed the number of observations, $n$.  In this case quantile regression will be inconsistent.  However, if the number of non-zero coefficients, $s = |\{(i,\ell): \beta_{i\ell} \neq 0\}|$, is sufficiently small, then we may recover them consistently using penalization.  In short, penalization shrinks coefficient estimates towards $0$ so that any non-zero estimates represent only the strongest lead-lag relationships among firms.
\medskip

\noindent \textbf{Optimization Problem. }
The exact penalty we apply can take different forms; here we choose an $\ell_1$ penalty known as the Least Absolute Shrinkage and Selection Operator (or Lasso) \citep{Lasso}.   
Our objective function for the $i$th firm then becomes
\begin{align}\label{eq:optimization}
\min_{\boldsymbol\beta \in \mathbb{R}^p}\left[\frac{1}{n}\sum_{t=1}^n\rho_\tau\left(x_{i,t+1}-\mathbf{x}_{t}'\boldsymbol\beta\right) + \lambda_i\sum_{\ell = 1}^p |\beta_{\ell}|\right],
\end{align}
where $\lambda_i$ is a tuning parameter that adjusts the amount of penalization.  When $\lambda_i = 0$, we recover our original (non-penalized) estimate, whereas in the limit as $\lambda_i \to \infty$, $\hat{\boldsymbol\beta}_i$ --- a minimizer of (\ref{eq:optimization}) --- will be identically zero.  

Equation (\ref{eq:optimization}) consists of two terms, the first of which is a weighted sum of residuals, and the second of which is a sum of the $\ell_1$ norm of each coefficient.  We can formulate (\ref{eq:optimization}) as a linear programming problem by introducing appropriate slack variables and considering the dual form \citep{koenker2014convex}.  We then minimize the result using the Barrodale-Roberts algorithm, which is a modification of the simplex method \citep{barrodalerobertsalgorithm}.  

\medskip

\noindent \textbf{Tuning Parameter Selection. } 
The optimal tuning parameter is selected using cross-validation, a method in which we divide our dataset into different folds, predict data in each fold using a model trained on all the other folds, and then select the value of $\lambda_i \in (0,\infty)$ that gives the best average performance across all folds.  Notice that this method may yield a different optimal tuning parameter, $\lambda^\ast_i$, for each firm $i$.  In practice, it is also possible to use the same tuning parameter for each firm, by setting $\lambda_i = \lambda := \frac{1}{p}\sum_{i=1}^p \lambda^\ast_i$ for all $i = 1,...,p$.

\section{Asymptotic Analysis}\label{sec:theory}

We investigate consistency and asymptotic normality of QGC estimators in the classical framework where the number of time series $p$ is fixed, and the sample size $n \rightarrow \infty$. The proof techniques build upon the asymptotic analyses of quantile regression for fixed design matrix \cite{quantileregressionKoenker}, quantile autoregression for univariate time series \citep{koenker_xiao}, and  Lasso estimators \citep{fu2000asymptotics}. To the best of our knowledge, asymptotic theory of Lasso-penalized quantile regressions for multivariate autoregressive design has not been investigated in the literature. We require additional assumptions on the multivariate autocovariance function of the time series to complete the proof. We start by laying out some notations before stating the assumptions and the result. \\

\textit{Notations}. We use $\mathcal{F}_t$ to denote the $\sigma$-field generated by the random variables $\{\mathbf{x}_1, \ldots, \mathbf{x}_t \}$. For a univariate time series $y_t$, we use  $Q_\tau(y_t|\mathcal{F}_t)$ to denote the conditional quantile of $y_t$ given the  $\mathcal{F}_t$. We also use $||\boldsymbol v||_1=\sum_{i=1}^p|v|_i$ to denote the $\ell_1$ norm of a p-dimensional vector $\boldsymbol v$. Convergence in distribution and probability will be denoted by $\Rightarrow$ and $\xrightarrow{\mathbb{P}}$ respectively. We use standard small o notation i.e. $a_n=o(b_n)$ to denote $a_n/b_n\rightarrow 0$. Similarly for random variables $X_n$ and $Y_n$ we use  $X_n=o_\mathbb{P}(Y_n)$ to denote $X_n/Y_n\xrightarrow{\mathbb{P}}0$. In particular $X_n=o_\mathbb{P}(1)$ is used to denote $X_n\xrightarrow{\mathbb{P}}0$.\\

We make the following assumptions on the multivariate centered and stationary time series $\{\mathbf{x}_t\}$.
\begin{enumerate}
\item[(A1)] Consider $i \in \{1, \ldots, p\}$, and a $\tau \in (0, 1)$. The univariate time series $y_t:= \mathbf{x}_{i,t+1}$ satisfies 
$$
y_t = \mathbf{x}_t'\beta^* + \xi_t
$$
where $\xi_t \stackrel{i.i.d.}{\sim} F$, with $F^{-1}(\tau) = 0$ and density $f = F'$ satisfying $f(0)>0$.

\item [(A2)] $\Omega_0:= \mathbb{E}[\mathbf{x}_t\mathbf{x}'_t]$  is invertible.

\item [(A3)] The autocovariance function 
$$
\gamma_{ij}(k):= \mathbb{E}[(\mathbf{x}_{i,t}\mathbf{x}_{j,t}-(\Omega_0)_{ij})(\mathbf{x}_{i,t-k}\mathbf{x}_{j,t-k}-(\Omega_0)_{ij})]
$$
satisfies $\sum_{k=1}^{\infty}|\gamma_{ij}(k)|<\infty$ for all $i,j = 1, \ldots, p$.
\end{enumerate}

Assumption (A1) implies that the conditional $\tau^{th}$ quantile of $x_{i,t+1}$ can be expressed as a linear combination of $\mathbf{x}_t$. This assumption should be taken as an approximation to study the asymptotic behavior of the univariate QR regressions for a given $(i, \tau)$. We are not assuming that exact linearity holds for every $\tau \in (0, 1)$. A formal theory under that generative model requires additional considerations about quantile path crossings \cite{koenker_xiao}, which is beyond the scope of this paper.

Assumption (A2) is standard in the literature of multivariate time series \cite{lutkepohl}. Assumption (A3) is needed to control the variance of the sample variance-covariance matrix, and is also common in the literature \cite{hamilton1994time}. Together, (A2) and (A3) ensures that the eigenvalues of $\frac{1}{n}\sum_{t=1}^n \mathbf{x}_t \mathbf{x}_t'$ remain bounded away from $0$ and $\infty$ asymptotically.

Our main result, presented below, states that with the right choice of tuning parameter, the regular and the Lasso penalized QGC estimators with multivariate autoregressive design are $\sqrt{n}$-consistent.

\begin{theorem}\label{prop:lowdim}
Consider a random realization $\{\mathbf{x}_1,\mathbf{x}_2,\ldots,\mathbf{x}_{n+1}\}$  from a centered and stationary process $\{\mathbf{x}_t \}_{t \ge 1}$ satisfying assumptions (A1)-(A3). Define 

$$\hat{\boldsymbol\beta}\in \argmin_{\boldsymbol\beta \in \mathbb{R}^p}\left[\frac{1}{n}\sum_{t=1}^{n}\rho_{\tau}(y_t-\mathbf{x}_t'\boldsymbol\beta)+\lambda_n||\boldsymbol\beta||_1\right].$$
Then for any $\lambda_n \ge 0$ with $\lambda_n=o(n^{-1/2})$, the estimator $\hat{\boldsymbol \beta}$ satisfies 

$$\sqrt{n}f(0)\Omega^{1/2}_0(\hat{\boldsymbol\beta}-\boldsymbol\beta^*)\Rightarrow N(0,\tau(1-\tau)I_p).$$
 
\end{theorem}
\begin{proof}[Proof of Theorem \ref{prop:lowdim}]
For any $\boldsymbol{\beta} \in \mathbb{R}^p$, define $\boldsymbol{v}=\sqrt{n}(\boldsymbol{\beta}-\boldsymbol{\beta}^*)$. Define the two  minimizers
\begin{align*}
    \hat{\boldsymbol\beta}&\in \argmin_{\boldsymbol\beta \in \mathbb{R}^p}\left[\frac{1}{n}\sum_{t=1}^{n}\rho_{\tau}(y_t-\mathbf{x}_t'\boldsymbol\beta)+\lambda_n||\boldsymbol\beta||_1\right],\\ \hat{\boldsymbol{v}}&\in \argmin_{\boldsymbol{v} \in \mathbb{R}^p}\left[\frac{1}{n}\sum_{t=1}^{n}\rho_{\tau}(\xi_t-n^{-1/2}\mathbf{x}_t'\boldsymbol{v})+\lambda_n||\boldsymbol\beta^*+\boldsymbol v/\sqrt{n}||_1\right].
\end{align*}

\noindent Note that $\hat{\boldsymbol{v}}=\sqrt{n}(\hat{\boldsymbol{\beta}}-\boldsymbol{\beta}^*)$.\\

\noindent We use the well-known Knight's identity \citep{quantileregressionKoenker}
\begin{align*}
    \rho_{\tau}(u-v)-\rho_{\tau}(u)=-v(\tau-\mathds{1}\{u<0\})+\int_0^v(\mathds{1}\{u<z\}-\mathds{1}\{u<0\})dz
\end{align*}
to write
\begin{align*}
Z_n(\boldsymbol v)&=\sum_{t=1}^n[\rho_{\tau}(\xi_t-n^{-1/2}\mathbf{x}_t'\boldsymbol{v})-\rho_{\tau}(\xi_t)]\\
&=\sum_{t=1}^n-n^{-1/2}\mathbf{x}_t'\boldsymbol{v}(\tau-\mathds{1}\{\xi_t<0\})+\sum_{t=1}^n\int_0^{n^{-1/2}\mathbf{x}_t'\boldsymbol{v}}(\mathds{1}\{\xi_t<z\}-\mathds{1}\{\xi_t<0\})dz.
\end{align*}

\noindent Let $\gamma^n_t(\boldsymbol v)=\int_0^{n^{-1/2}\mathbf{x}_t'\boldsymbol{v}}(\mathds{1}\{\xi_t<z\}-\mathds{1}\{\xi_t<0\})dz$ and $\bar{\gamma}^n_t(\boldsymbol v)=\mathbb{E}[\gamma^n_t(\boldsymbol v)|\mathcal{F}_{t}]$. Also define
$$
W_n(\boldsymbol v)=\sum_{t=1}^n\gamma^n_t(\boldsymbol v), \quad \bar{W}_n(\boldsymbol v)=\sum_{t=1}^n\bar{\gamma}^n_t(\boldsymbol v).
$$
We analyse the behaviour of the term $\bar{W}_n(\boldsymbol v)$.
\begin{align*}
  \bar{W}_n(\boldsymbol v)&=\sum_{t=1}^n\int_0^{n^{-1/2}\mathbf{x}_t'\boldsymbol{v}}(F(z)-F(0))dz\\
  &= \sum_{t=1}^n\int_0^{n^{-1/2}\mathbf{x}'_t\boldsymbol{v}}f(0)zdz+o_{\mathbb{P}}(1)\\
    &= \frac{f(0)}{2}\boldsymbol{v}'\left(\frac{1}{n}\sum_{t=1}^n\mathbf{x}_t\mathbf{x}'_t\right)\boldsymbol{v}+o_{\mathbb{P}}(1).
  \end{align*}
We first show that the first term converges to $\frac{f(0)}{2}\boldsymbol{v}'\Omega_0\boldsymbol{v}$ in probability. Indeed it is enough to show $\frac{1}{n}\sum_{t=1}^n\mathbf{x}_{i,t}\mathbf{x}_{j,t}\xrightarrow{\mathbb{P}}(\Omega_0)_{ij}$.
From the covariance stationarity we have 
$\mathbb{E}[\mathbf{x}_{i,t}\mathbf{x}_{j,t}]=(\Omega_0)_{ij}$. Since 
$$
\mathbb{E}[(\mathbf{x}_{i,t}\mathbf{x}_{j,t}-(\Omega_0)_{ij})(\mathbf{x}_{i,t-k}\mathbf{x}_{j,t-k}-(\Omega_0)_{ij})]=\gamma_{ij}(k)
$$
with $\sum_{k=1}^{\infty}|\gamma_{ij}(k)|<\infty$, we can use Proposition 7.5 of \citet{hamilton1994time} with $Y_t=\mathbf{x}_{i,t}\mathbf{x}_{j,t}$ to obtain $\frac{1}{n}\sum_{t=1}^n\mathbf{x}_{i,t}\mathbf{x}_{j,t}\xrightarrow{\mathbb{P}}(\Omega_0)_{ij}$.
Thus we obtain
$$\bar{W}_n(\boldsymbol v)\xrightarrow{\mathbb{P}}\frac{f(0)}{2}\boldsymbol{v}'\Omega_0\boldsymbol{v}.$$
We observe that $\gamma^n_t(\boldsymbol v)- \bar{\gamma}^n_t(\boldsymbol v)$ constitute a Martingale difference sequence and using an argument similar to the one used in proving Theorem 1 in \citet{herce}, we can show that $W_n(\boldsymbol v)- \bar{W}_n(\boldsymbol v)\xrightarrow{\mathbb{P}}0$ and hence $W_n(\boldsymbol v)\xrightarrow{\mathbb{P}}\frac{f(0)}{2}\boldsymbol{v}'\Omega_0\boldsymbol{v}$.\\

\noindent Let $U_n(\boldsymbol v)=-\sum_{t=1}^n \left[n^{-1/2}\mathbf{x}'_t\boldsymbol{v}(\tau-\mathbbm{1}(\xi_t<0))\right]$ and observe that
$$
\mathbb{E}[\mathbf{x}'_t\boldsymbol{v}(\tau-\mathbbm{1}(\xi_t<0))|\mathcal{F}_{t}]=0.
$$
Thus we can use Martingale CLT \citep{hall1980martingale} to conclude that
$$
U_n(\boldsymbol v)\Rightarrow N(0,\tau(1-\tau)\boldsymbol{v}'\Omega_0\boldsymbol{v}).
$$

\noindent Let $Z(\boldsymbol v)=\frac{f(0)}{2}\boldsymbol{v}'\Omega_0\boldsymbol{v}-\boldsymbol{v}'\Omega_0^{1/2}Z$, where $Z \sim N(0,\tau(1-\tau)I_p)$, and define

$$\tilde{Z}_n(\boldsymbol v)=Z_n(\boldsymbol v)+n \lambda_n(||\boldsymbol\beta^*+\boldsymbol v/\sqrt{n}||_1-||\boldsymbol\beta^*||_1).$$

\noindent It can be seen that if $\lambda_n=o(n^{-1/2})$, then
$$
\tilde{Z}_n(\boldsymbol v)\Rightarrow Z(\boldsymbol v).
$$ 
Also it follows that the minimizer of $Z(\boldsymbol v)$
is $f(0)^{-1}\Omega_0^{-1/2}Z$. Now we adopt a standard convexity argument from \citet{pollard1991asymptotics}, which has also been used in the analyses of univariate QAR processes in \citet{koenker_xiao} and Lasso penalized quantile regression for fixed design in \citet{wu2009variable}. The convexity argument ensures that since $\tilde{Z}_n(\boldsymbol v)$ converges to $Z(\boldsymbol v)$ as a process, the minimizer of $\tilde{Z}_n(\boldsymbol v)$ converges to that of $Z(\boldsymbol v)$ in distribution. Thus we obtain

$$\hat{\boldsymbol{v}}=\sqrt{n}(\hat{\boldsymbol{\beta}}-\boldsymbol{\beta}^*)\Rightarrow f(0)^{-1}\Omega_0^{-1/2}Z.$$
Rearranging the terms, we obtain 
$$\sqrt{n}f(0)\Omega^{1/2}_0(\hat{\boldsymbol\beta}-\boldsymbol\beta^*)\Rightarrow N(0,\tau(1-\tau)I_p),$$
proving the desired result.
\end{proof}

\section{Numerical Experiments}\label{sec:simulation}
In this section, we perform a simulation study to compare how accurately the networks based on multivariate GC and QGC detect lower tail linkages.  Our simulation study is based on a hub network in which there exist linkages between the hub node and the nodes connected to it \textit{only in the lower tail} of the nodes' return  distribution.  We generate time series data according to a model with the given network structure, and compare the two methods' sensitivity and specificity.  We report results aggregated over 50 replicates.

We consider networks with $p \in \{30, 70\}$ nodes, corresponding -- approximately -- to the size of the empirical U.S. and India networks that we estimate in Sections \ref{sec:US_empirical_section} and \ref{sec:India_empirical_section}, respectively.  Each simulated network contains $\frac{p}{10}$ connected components of 10 nodes each, one node acting as the hub and the rest serving as peripheral nodes (see Figure \ref{fig:hub_network_schematic} for schematic).  We generate time series for each node according to the following model.  For a hub node, $h$, and a peripheral node, $x$, that is connected to $h$, we take the values of the time series at time $t$ to be
\begin{align}
\label{eq:hubTS}
    h_t &= \begin{cases}
    0.4h_{t-1} + \epsilon^{O,h}_t & \text{if } f_t = 0 \\
    \epsilon^{B,h}_t & \text{if } f_t = 1
    \end{cases}\\
    \label{eq:otherTS}
    x_t &= \begin{cases}
    0.4x_{t-1} + \epsilon^{O,p}_t & \text{if } f_{t-1} = 0 \\
    0.4x_{t-1} + 0.6h_{t-1} + \epsilon^{B,p}_t & \text{if } f_{t-1} = 1
    \end{cases}
\end{align}
where $\epsilon^{O,h}_t$, $\epsilon^{O,p}_t$, and $\epsilon^{B,p}_t$ are independent and identically distributed $\mathcal{N}(\mu = 0, \sigma = 0.1)$ random variables, $\epsilon^{B,h}_t \sim \mathcal{N}(\mu = -0.8, \sigma = 0.1)$, and $f_t \sim Bernoulli(0.05)$.  Superscript $O$ (resp. $B$) denotes ``ordinary'' (resp. ``bad'') time points, while superscript $h$ (resp. $p$) stands for ``hub'' (resp. ``peripheral'') node.  For example, $\epsilon^{O,h}_t$ is the error term for a hub node at an ``ordinary'' time point.  The Bernoulli random variable, $f_t$, represents an exogenous factor, such as the state of the economy at time $t$.  When the economy declines ($f_t = 1$), the hub node's time series descends to a low value (see equation \eqref{eq:hubTS}).  This is captured by the fact that $\epsilon^{B,h}_t$ has mean $-0.8$, which is approximately the 20$^{th}$ percentile of U.S. stock returns in our data.\footnote{As detailed in Section \ref{sec:US_empirical_section}, we consider 36-month rolling windows of historical U.S. stock returns.  For each of these windows, we computed the 20th percentile of the returns, and then calculated the average of these values over all windows to yield a number that was approximately -0.8.}  At the following time point, the time series of the peripheral nodes that are connected to the hub may also descend: from equation \eqref{eq:otherTS}, we see that the value of $x_t$ depends on its own past, $x_{t-1}$, and on the hub node's past, $h_{t-1}$.  On the other hand, when the economy is doing well ($f_t = 0$), both the hub and peripheral nodes evolve according to an $AR(1)$ process.  

\begin{figure}[ht]
\centering
\begin{subfigure}{.25\textwidth}
\includegraphics[scale=0.3, trim = {2cm 2cm 3cm 2cm}, clip]{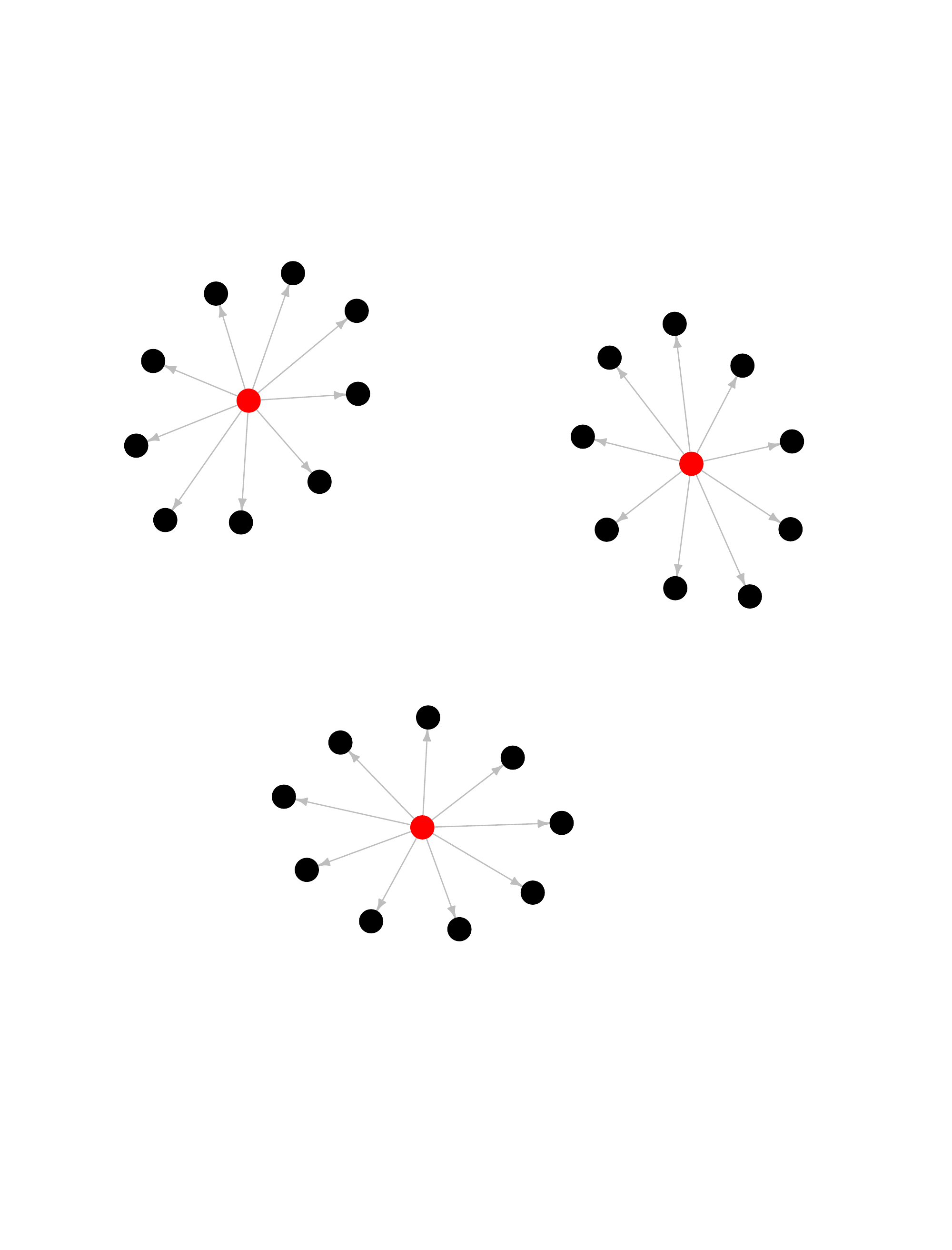}
\end{subfigure}
\hspace{0.5in}
  \begin{subfigure}{.25\textwidth}
  \includegraphics[angle = -90, scale=0.26, trim = {0cm 0cm 0cm 0cm}, clip]{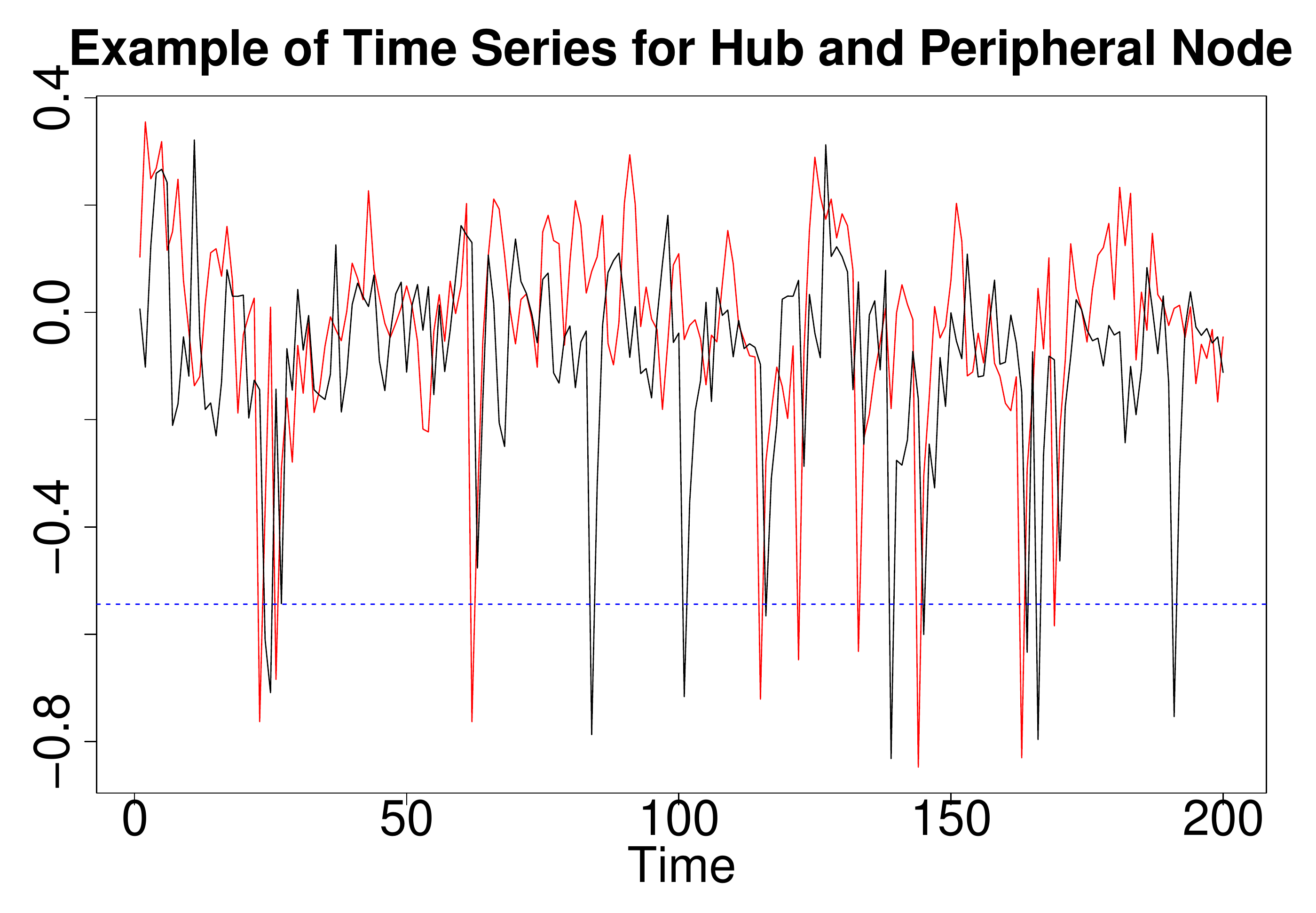}
  \end{subfigure}
    \caption{[Left]: a hub network of $p = 30$ nodes, consisting of 3 connected components with 10 nodes each.  The 3 hub nodes are indicated in red, while all non-hub (i.e., peripheral) nodes are black.  [Right]: sample time series generated according to equations \eqref{eq:hubTS} and \eqref{eq:otherTS}.  The solid red (resp. black) line depicts the time series for one of the hub (resp. peripheral) nodes, with the dashed blue line marking the empirical 5$^{th}$ percentile of the peripheral node's time series.  Note that, by construction, the value of the peripheral node's time series often drops below its 5$^{th}$ percentile following a drop in the hub node's time series.  
    }
    \label{fig:hub_network_schematic}
\end{figure}

Using equations \eqref{eq:hubTS} and \eqref{eq:otherTS}, we generate time series of length $n \in \{25, 50, 75, 100\}$\footnote{For each $n$, we use a burn-in period of length 500.} and build networks based on Lasso penalized GC and QGC ($\tau = 0.05$) on our simulated datasets.  In each case, we use 10-fold cross-validation to compute the optimal penalization tuning parameter for each node.  

Table \ref{table: sens_spec_hub_ntwk} shows the sensitivity and specificity ( aggregated over 50 replicates) of the two methods for various $(n,p)$ combinations.  For all values of $n$ and $p$, QGC yields higher specificity (i.e., fewer false positives) than does GC, and two methods have comparable sensitivity (within one standard deviation of each other).  Taken together, these results suggest that quantile regression can successfully detect tail linkages while avoiding many of the false positives produced by mean-based methods. Figure \ref{fig:hub_ntwk_proportion} further illustrates how QGC produces fewer spurious edges than GC.  

\begin{table}[ht]
\centering
\setlength{\tabcolsep}{7pt}
\begin{tabular}{c|cccccccc}
\toprule
& \multicolumn{4}{c}{$p = 70, |E| = 63$} & \multicolumn{4}{c}{$p = 30, |E| = 27$} \\ 
& \multicolumn{2}{c}{Sensitivity} & \multicolumn{2}{c}{Specificity} &  \multicolumn{2}{c}{Sensitivity} & \multicolumn{2}{c}{Specificity} \\
$n$ & GC & QR 0.05 & GC & QR 0.05 & GC & QR 0.05 & GC & QR 0.05
\\ \midrule
25 & 62 (14) & 55 (14) & 91 (1) & 98 (0) &  65 (20)	& 56 (22)&	82 (3)&	97 (1)\\
50 & 85 (10) & 81 (13) & 86 (2) & 98 (0) &   86 (12)&	83 (15)	&81 (4)	&95 (2) \\
75 & 94 (7) & 92 (8) & 85 (2) & 97 (0) &   94 (11)&	93 (12)	&79 (4)	&93 (2) \\
100 & 99 (3) & 97 (4) & 85 (2) & 96 (1)  &  97 (6)	&98 (6)	&78 (4) &91 (2) \\
\bottomrule
\end{tabular}
\caption{Sensitivity and specificity for hub network recovery using networks based on GC and QGC ($\tau = 0.05$).  Mean sensitivity and specificity (averaged over 50 replicates) are expressed as percentages, with standard deviations in parentheses. \label{table: sens_spec_hub_ntwk}}
\end{table}

\begin{figure}[h]
    \centering
    \includegraphics[scale = 0.85, trim = {2.3cm 0.9cm 1cm 0.9cm}, clip = TRUE]{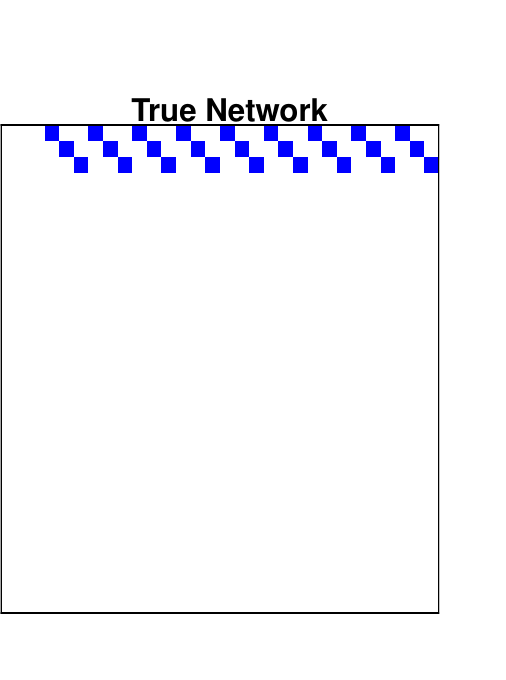}
    \includegraphics[scale = 0.35, trim = {0.4cm 0.9cm 0.9cm 0cm}, clip = TRUE]{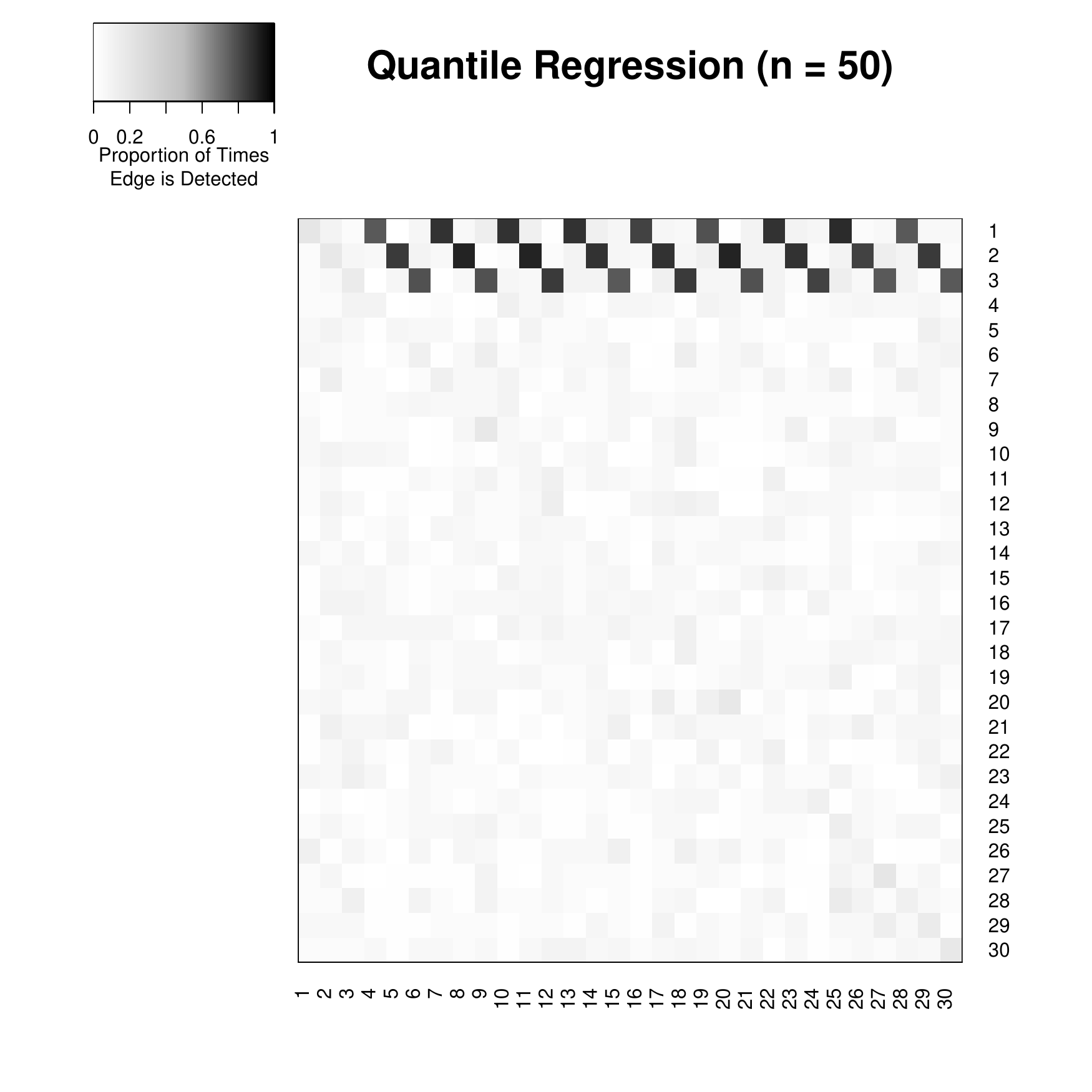}
    \includegraphics[scale = 0.35, trim = {0.4cm 0.9cm 0.9cm 0cm}, clip = TRUE]{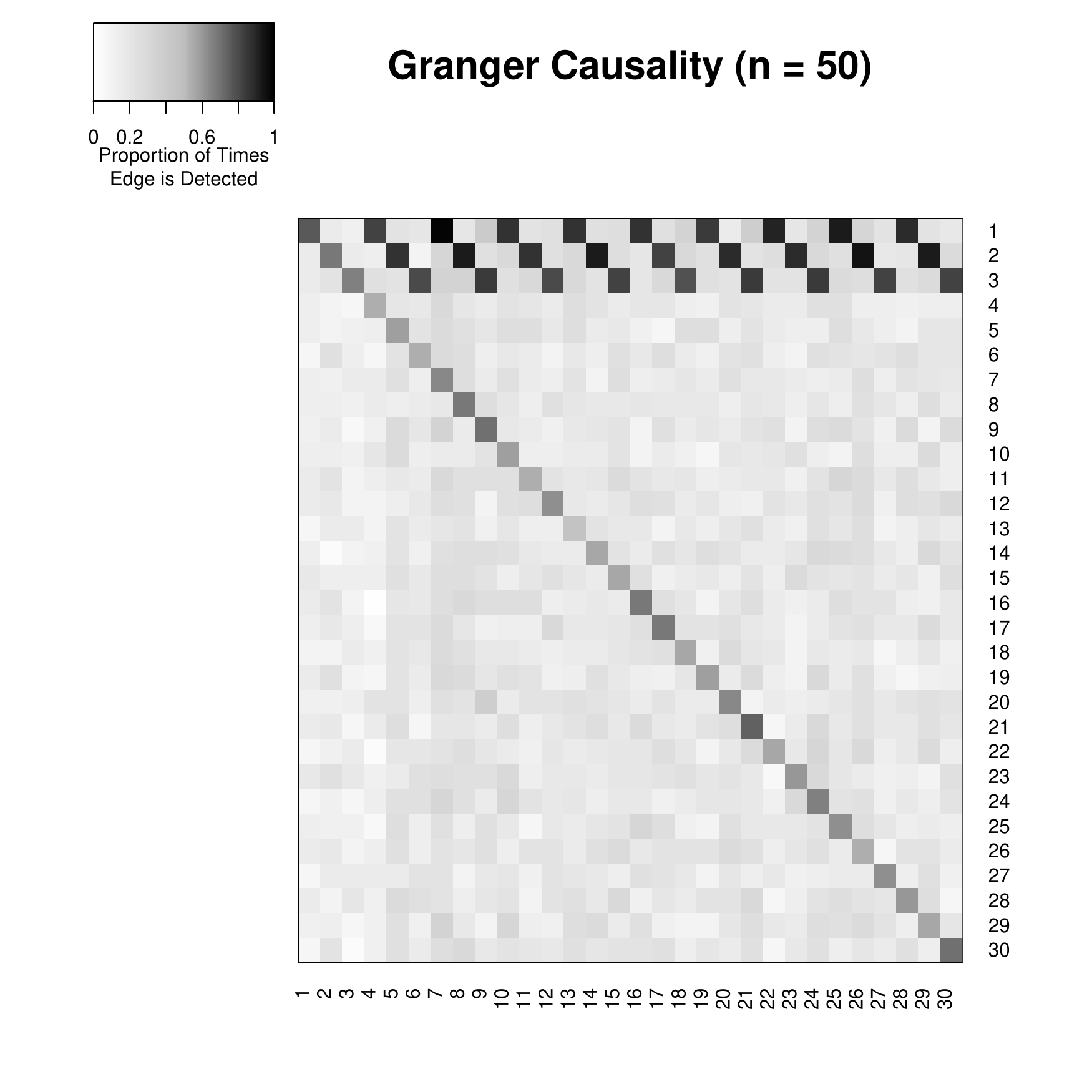}
    \caption{Heatmaps for edge detection based on simulated data from hub networks.  [Left]: the binary adjacency matrix for a hub network with $p = 30$ nodes.  Nodes 1-3 are the hubs and blue squares indicate presence of an edge.  [Middle]: the proportion of times -- out of 50 experiments -- that the edge is detected by QGC ($\tau = 0.05$). [Right]: the proportion of times -- out of 50 experiments -- that the edge is detected by GC.}
    \label{fig:hub_ntwk_proportion}
\end{figure}
\section{Empirical Results:  U.S. Financial Firms}
\label{sec:US_empirical_section}

{We now apply the above methods to historical stock returns. After describing the data collection, preprocessing, and network formation details, we illustrate how well the network centrality measures can detect bouts of financial instability during the period of 1990-2012. We also focus on the Financial Crisis of 2007-2009 and take a closer look at the most central firms during this time.}

\medskip

\noindent {\textbf{Data Collection and Pre-processing}}. 
We collected stock price data from the University of Chicago's Center for Research in Security Prices (CRSP) using the Wharton Research Data Services \citep{CRSP}. 
 
Following \citet{billio2012econometric} and \citet{preprint}, the original data set contains monthly stock returns from January 1990 to December 2012 on firms in three financial sectors: banks (BA), primary broker/dealers (PB), and insurance companies (INS).  Sectoral membership of firms can be identified using the Standard Industrial Classification (SIC) code, a 4-digit code used by the U.S. government to classify industries.  

We divide the time period into 36-month rolling windows.  For each window, we rank the firms according to their average market capitalization, and discard all but the top 25 firms from each sector.  For each window, this leads to a data matrix of raw returns $\mathbf{X} \in \mathbb{R}^{(n+1) \times p}$, where $n+1 = 36$ and $p = 75$.  

Following \citet{billio2012econometric}, we used a GARCH(1,1) filter \citep{GARCH} to remove heteroskedasticity in raw returns. 
 
GARCH(1,1) models the standard deviation of $x_{i,t}$ on the lagged value of the series, $x_{i,t-1}$, as well as the series' lagged standard deviation.  Specifically, we assume
\begin{align*}
x_{i,t} &= \mu_i + \sigma_{i,t}\epsilon_{i,t} \hspace{10mm}\epsilon_{i,t} \sim \mathcal{N}(0,1) \\
\sigma_{i,t}^2 &= \omega_i + \gamma_i(x_{i,t-1} - \mu_i)^2 + \eta_i\sigma_{i,t-1}^2,
\end{align*}
where $\omega_i, \gamma_i,$ and $\eta_i$ are constants.
Thus, an extreme return and/or high volatility at time $t-1$ is propagated to the series at time $t$.  
Given a series $\left\{x_{i,t}\right\}_{t=1}^{n+1}$, we compute the maximum likelihood estimates for $\mu_i, \omega_i, \gamma_i$, and $\eta_i$ and then form the standardized residuals
\begin{equation*}
\hat{\epsilon}_{i,t} := \frac{x_{i,t} - \hat{\mu}_i}{\hat{\sigma}_{i,t}} \overset{\text{approx}}\sim \mathcal{N}(0, 1) \hspace{10mm} t = 1,...,(n+1).
\end{equation*}
All of our GC network models are fitted on these residuals  $\left\{\hat{\epsilon}_{i,t}\right\}_{t=1}^{n+1}$, for $i = 1,...,p$.  

\medskip
\noindent {\textbf{Network Estimation and Summary Statistics}}. We build financial networks based on GC and QGC ($\tau$ = $0.05,\, 0.9$) for each of the 36-month rolling windows spanning 1990 to 2012.   While QGC($\tau=0.05$) identifies linkages in the lower tail of the stock returns distribution, QGC($\tau=0.9$) does the same for the upper tail, corresponding to firms' best-performing days.

In order to study the evolution of network connectedness over time and detect central firms, we used degree centrality to summarize the network information. The degree of node $i$ is defined as the number of edges adjacent to $i$; that is, how many direct neighbors $i$ has. The degree of the network is defined as the average degree over all nodes in the network. Our methods produced qualitatively similar results when we used closeness centrality, so we only report the results based on degree.

\subsection{Evolution of Network Summary Statistics (1990-2012)}

In what follows, we demonstrate how well our network centrality measure (average degree) can detect financial crisis periods.  In the left column of Figure \ref{fig:time_series_US_networks}, we plot the average network degree for each 36-month window.  These networks were constructed using bivariate Granger causality (GC), as well as bivariate quantile regression with $\tau = 0.05$ (QR-0.05) and $\tau = 0.9$ (QR-0.9).  In each window, we scale degree by its historical average (over all rolling windows) so that we can make valid comparisons across the three estimation procedures. 

Our first finding is that average degree often increases before or during systemic events.  For example, bivariate Granger causality networks display increased connectivity during the 1998 Russian default and the subsequent collapse of the hedge fund, Long Term Capital Management (LTCM), the 2008 bankruptcy of Lehman Brothers (considered a critical event during the 2007-2009 U.S. Financial Crisis), and the U.S. debt-ceiling crisis of 2011.  
 
We observe the same overall pattern in the networks that were estimated using quantile regression -- with some notable differences.  First, quantile regression with $\tau = 0.9$ does not clearly identify the US Financial Crisis; there is very little increase in connectivity during this period.  However, the lower tail networks display a prominent increase in average degree during this time period, and -- interestingly -- the increase occurs earlier than it does in Granger causality networks.  QR-0.05 networks become dense starting around the beginning of 2008, whereas the Granger causality networks do not reach peak connectivity until fall of that year.  Second, QR-0.05 networks have high average degree around the 2000 dot-com bubble, while GC networks do not.  Thus, by considering interlinkages in the lower tails of the returns distribution, we may be able to identify events that are not discernible by looking at the center of the distribution. 

In the right column of Figure \ref{fig:time_series_US_networks}, we repeat the above plots for networks that have been estimated using the multivariate approach.  As before, Granger causality and lower-tail quantile regression networks are dense around several systemic events, including the 1998 Russian default, the 2007-2009 Financial Crisis, and the 2011 U.S. debt ceiling crisis.  The upper tail networks do not display an increase in average degree during many of these periods, suggesting that much of the connectivity in mean-based (Granger causality) networks may be driven by connections on bad days of the market. \\ 

\begin{figure}[h]
\centering
\begin{subfigure}{0.45\textwidth}
    \includegraphics[scale=0.25,angle=-90]{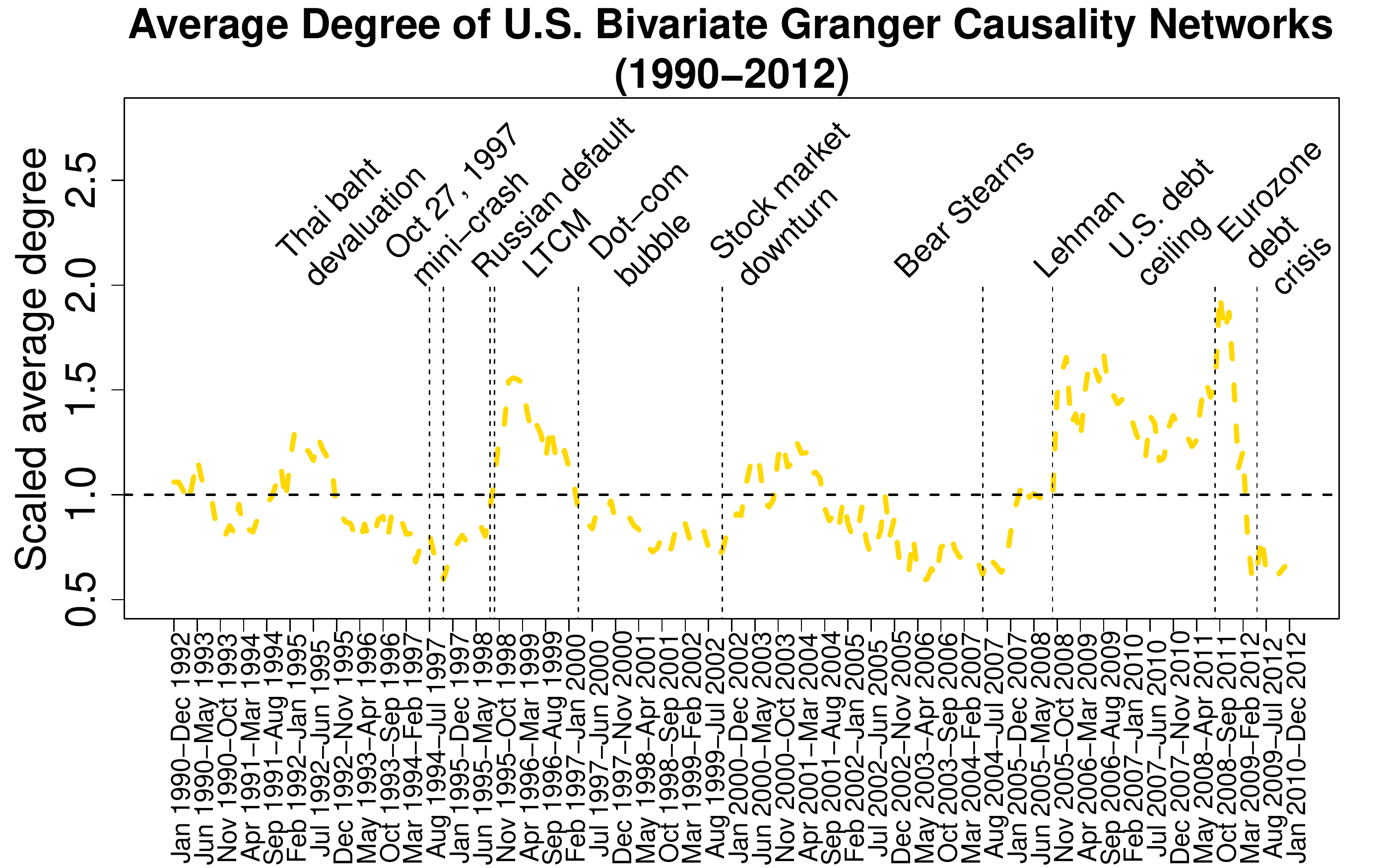}
\end{subfigure}
\begin{subfigure}{0.45\textwidth}
    \includegraphics[scale=0.25,angle=-90]{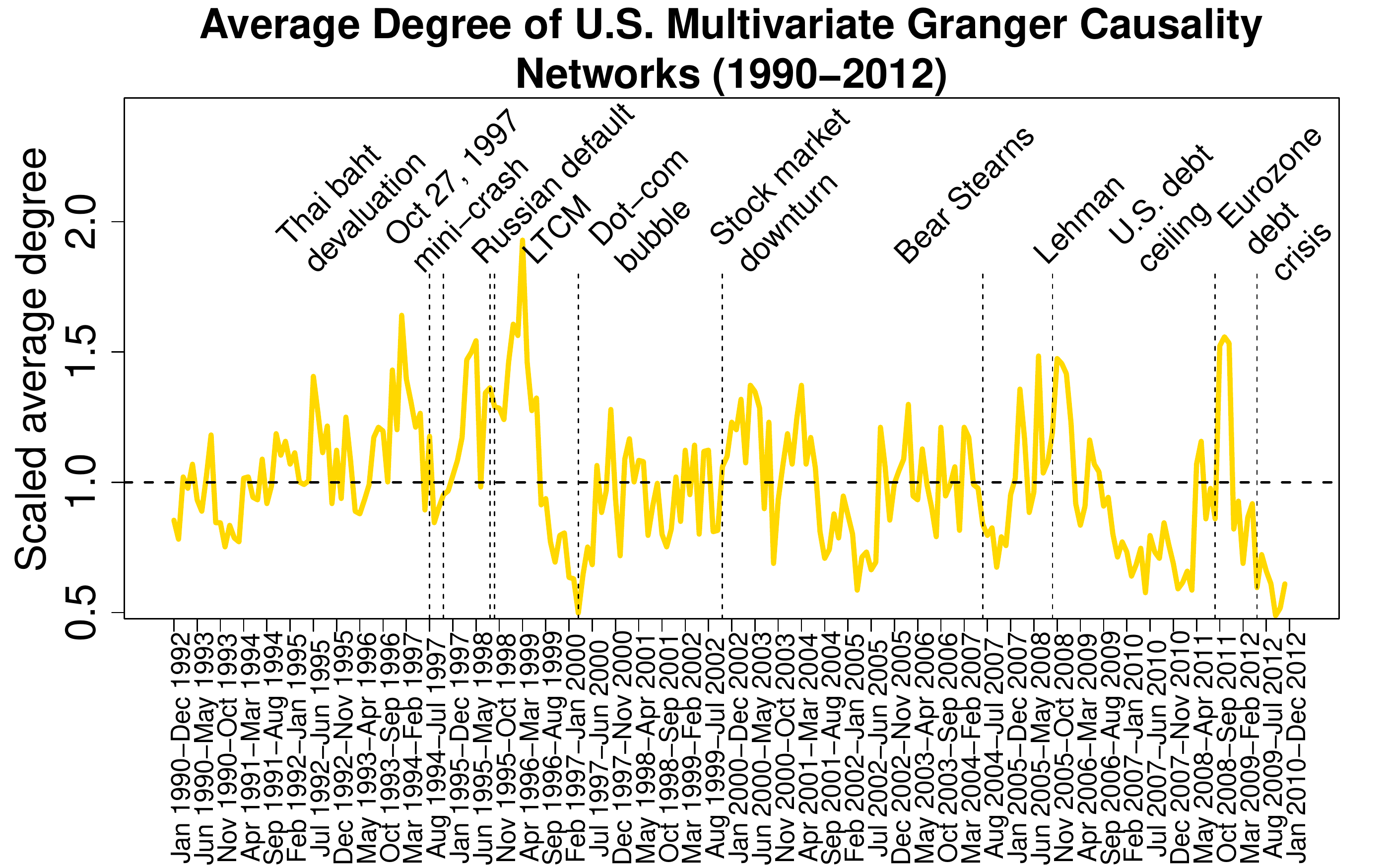}
\end{subfigure}
\vskip\baselineskip
\begin{subfigure}{0.45\textwidth}
    \includegraphics[scale=0.25,angle=-90]{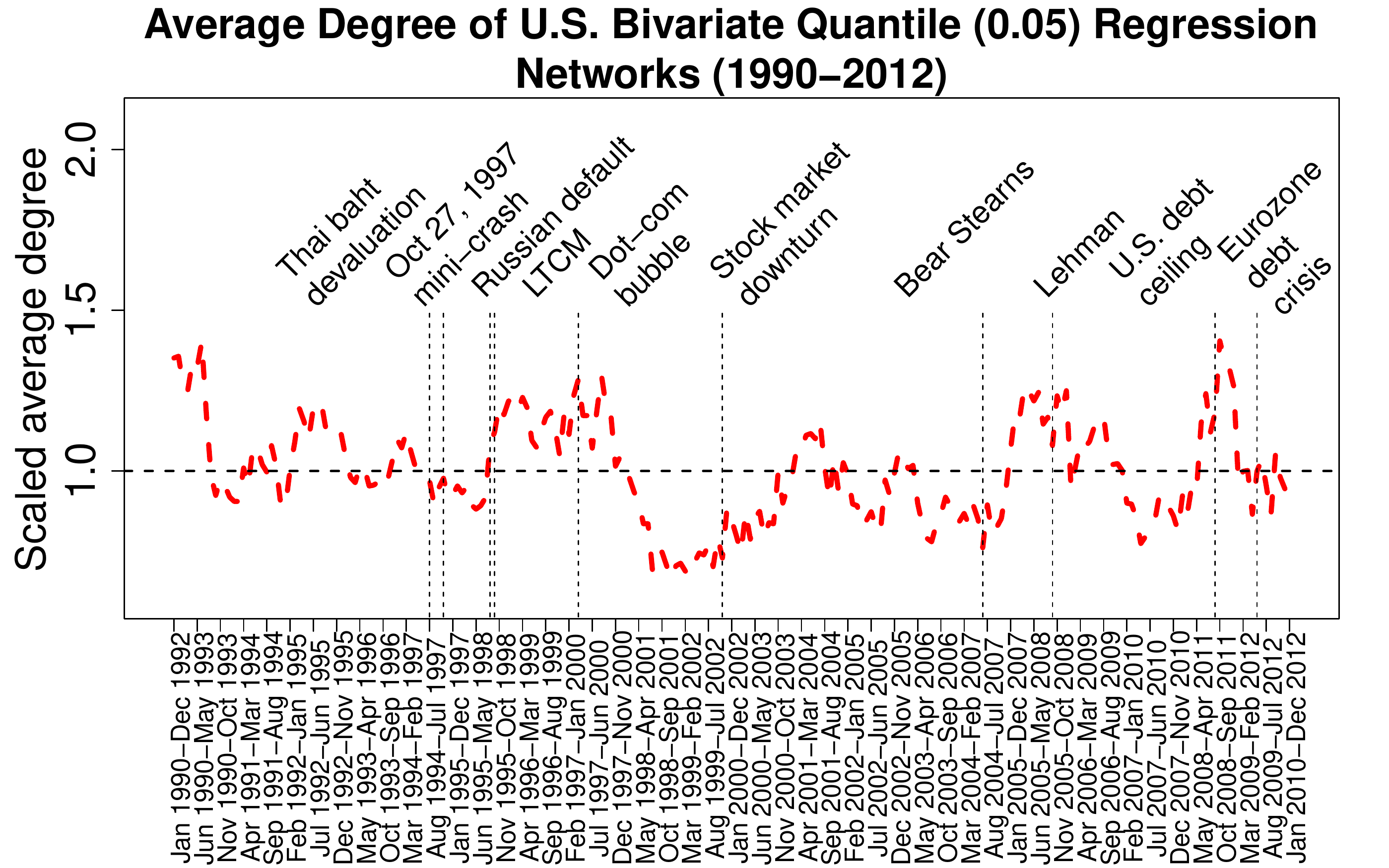}
\end{subfigure}
\begin{subfigure}{0.45\textwidth}
    \includegraphics[scale=0.25,angle=-90]{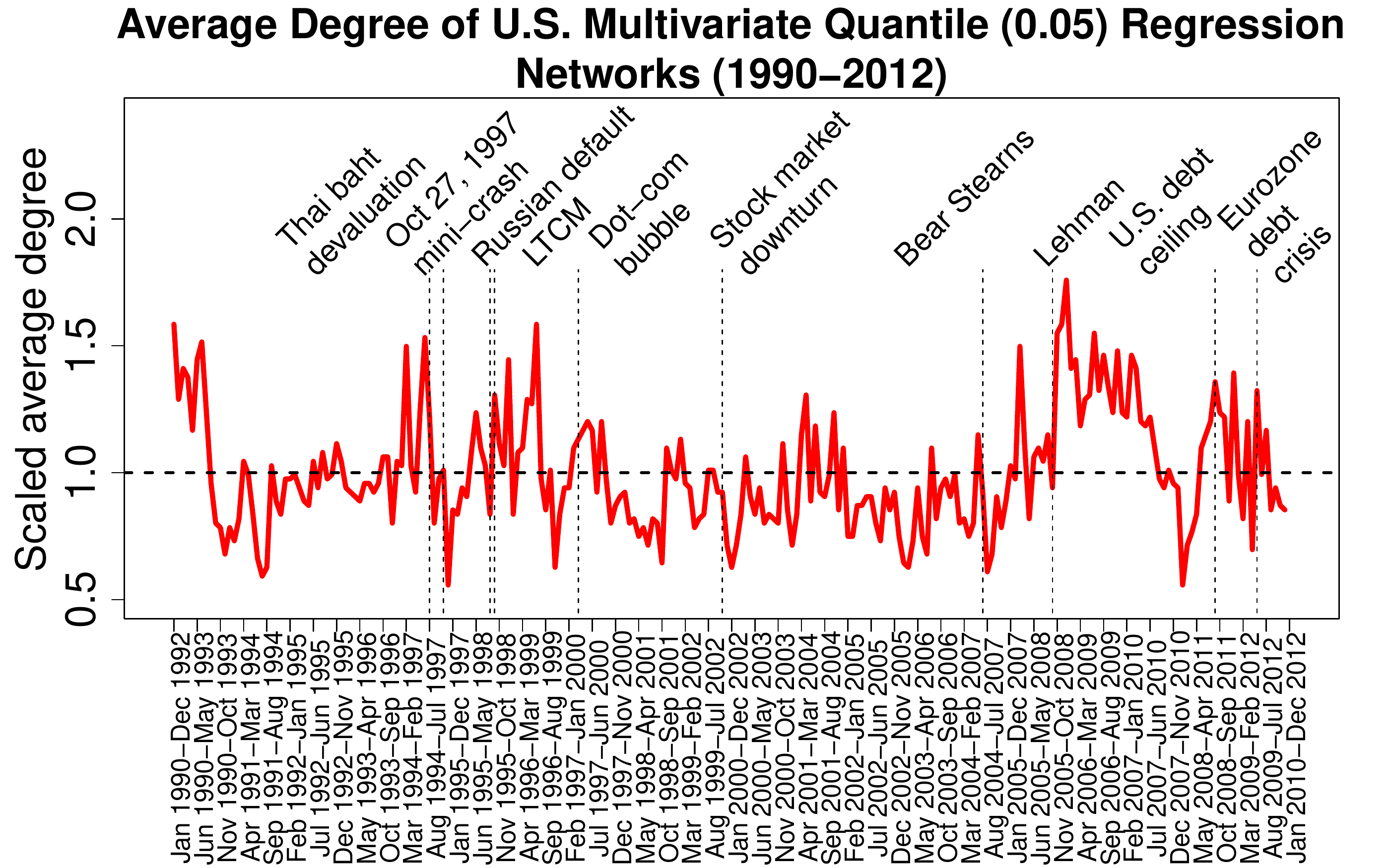}
\end{subfigure}
\vskip\baselineskip
\begin{subfigure}{0.45\textwidth}
    \includegraphics[scale=0.25,angle=-90]{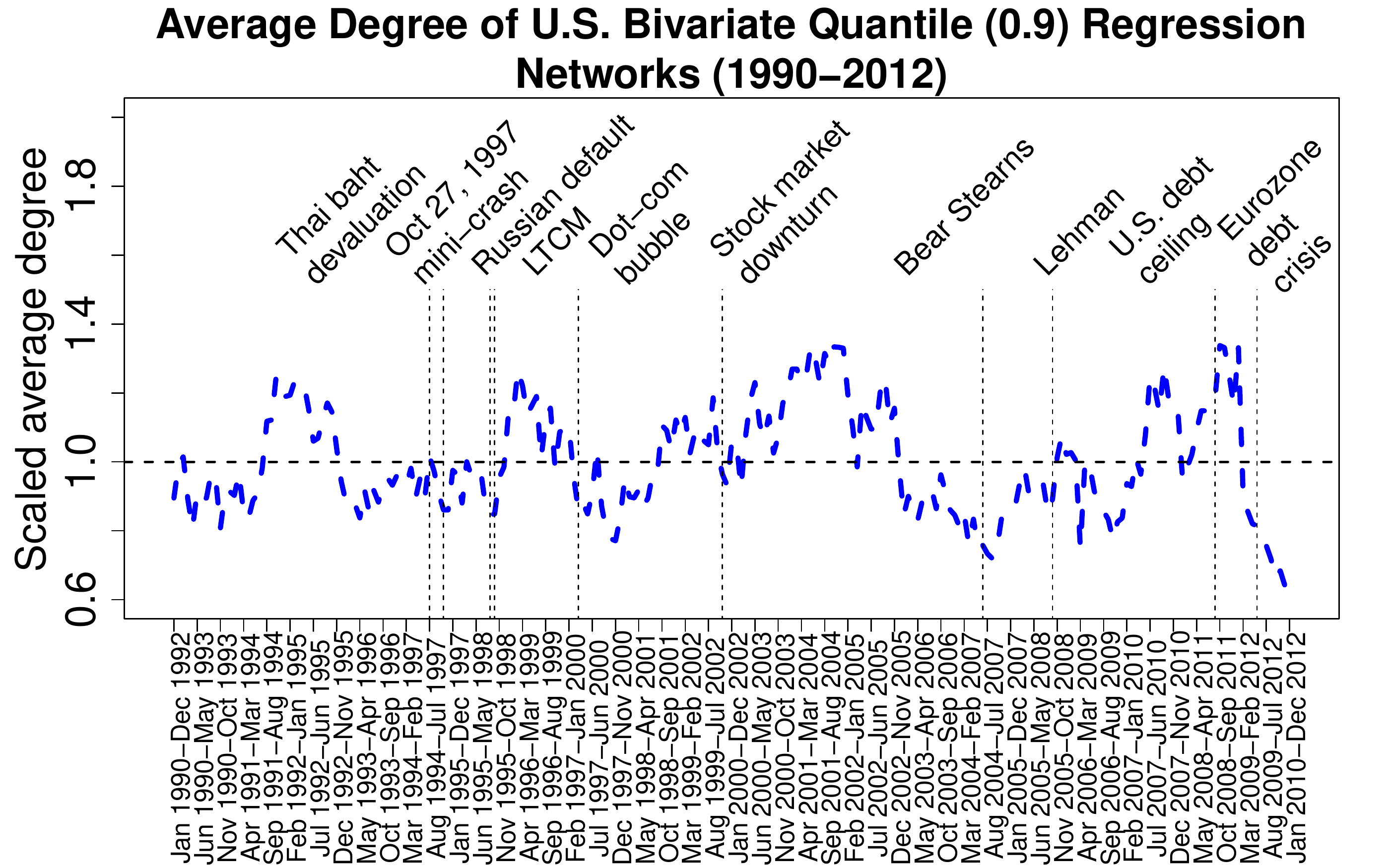}
\end{subfigure}
\begin{subfigure}{0.45\textwidth}
    \includegraphics[scale=0.25,angle=-90]{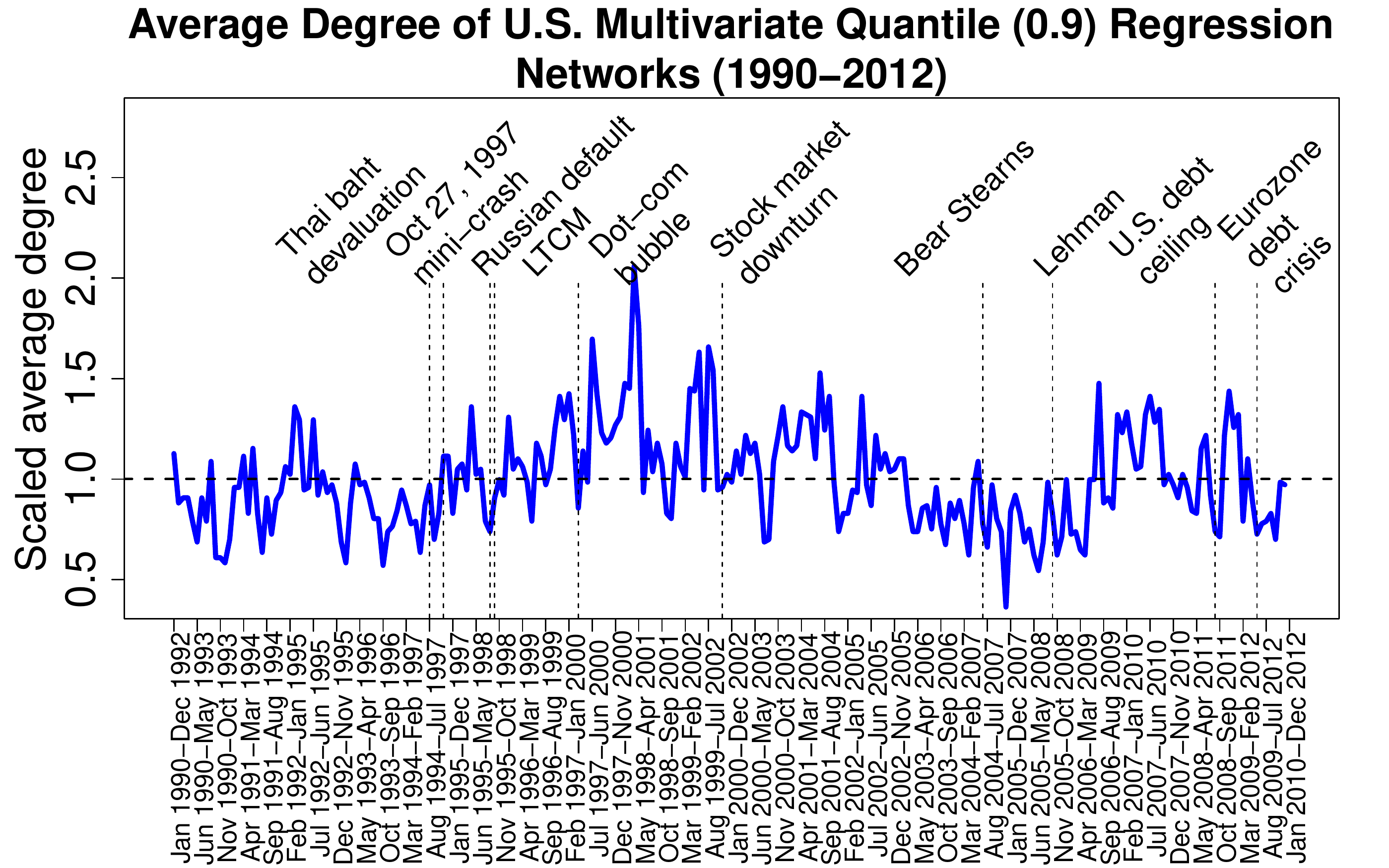}
\end{subfigure}
    \caption{Scaled average degree of networks estimated using GC (top), QR-0.05 (middle), and QR-0.9 (bottom).  [Left]: bivariate results.  [Right]: multivariate results.  The scaled average degree is computed by dividing the mean degree for each network by the historical mean degree (i.e., over the entire sample period).  For GC and QR-0.05, average degree tends to increase before and/or during systemic events.}
    \label{fig:time_series_US_networks}
\end{figure}

\begin{figure}[!t]
    \centering
    \begin{subfigure}{0.45\textwidth}
    \includegraphics[scale=0.22]{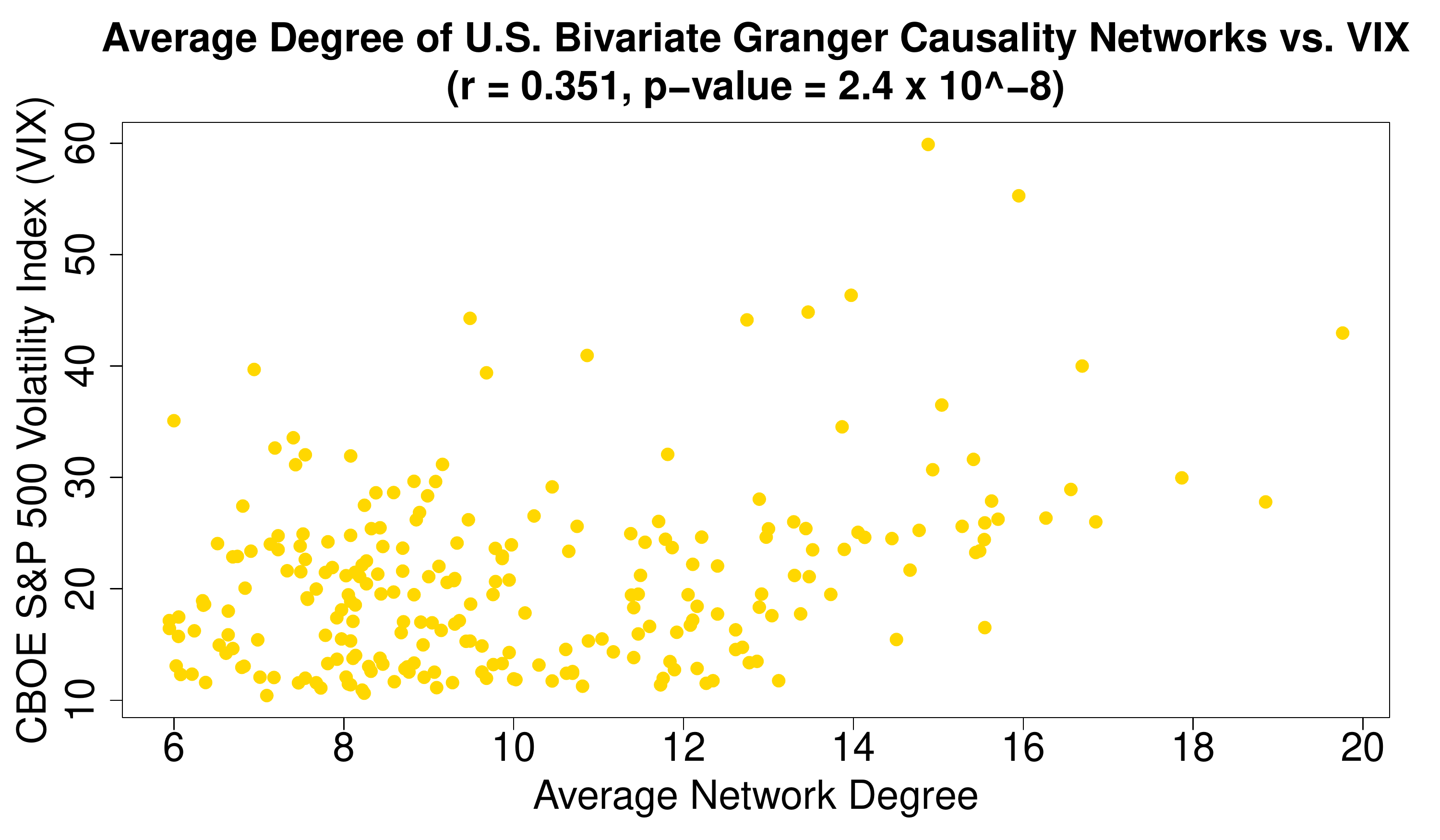}
    \end{subfigure}
    \begin{subfigure}{0.45\textwidth}
        \includegraphics[scale=0.22]{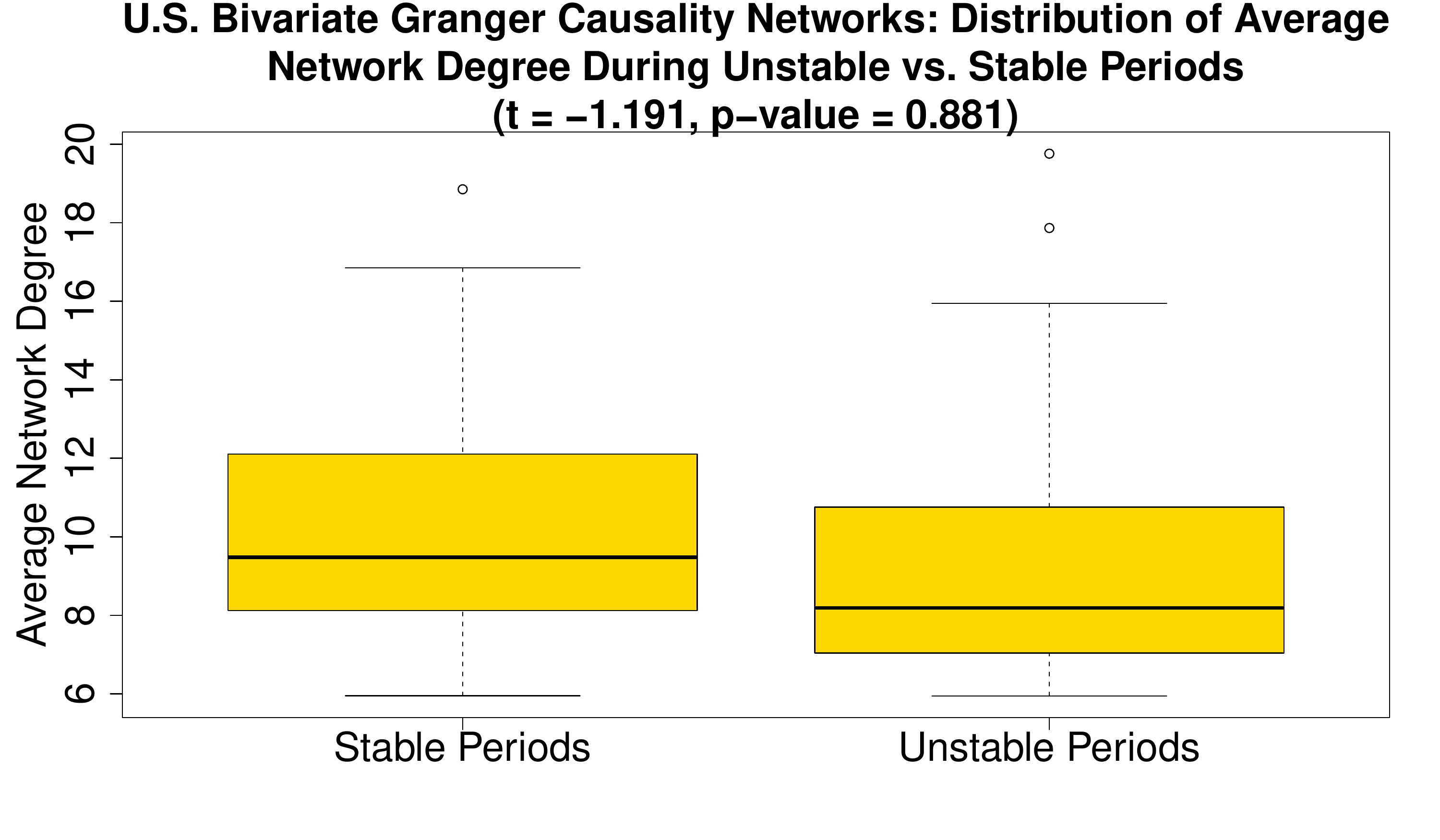}
    \end{subfigure}
    \vskip\baselineskip
    \begin{subfigure}{0.45\textwidth}
        \includegraphics[scale=0.22]{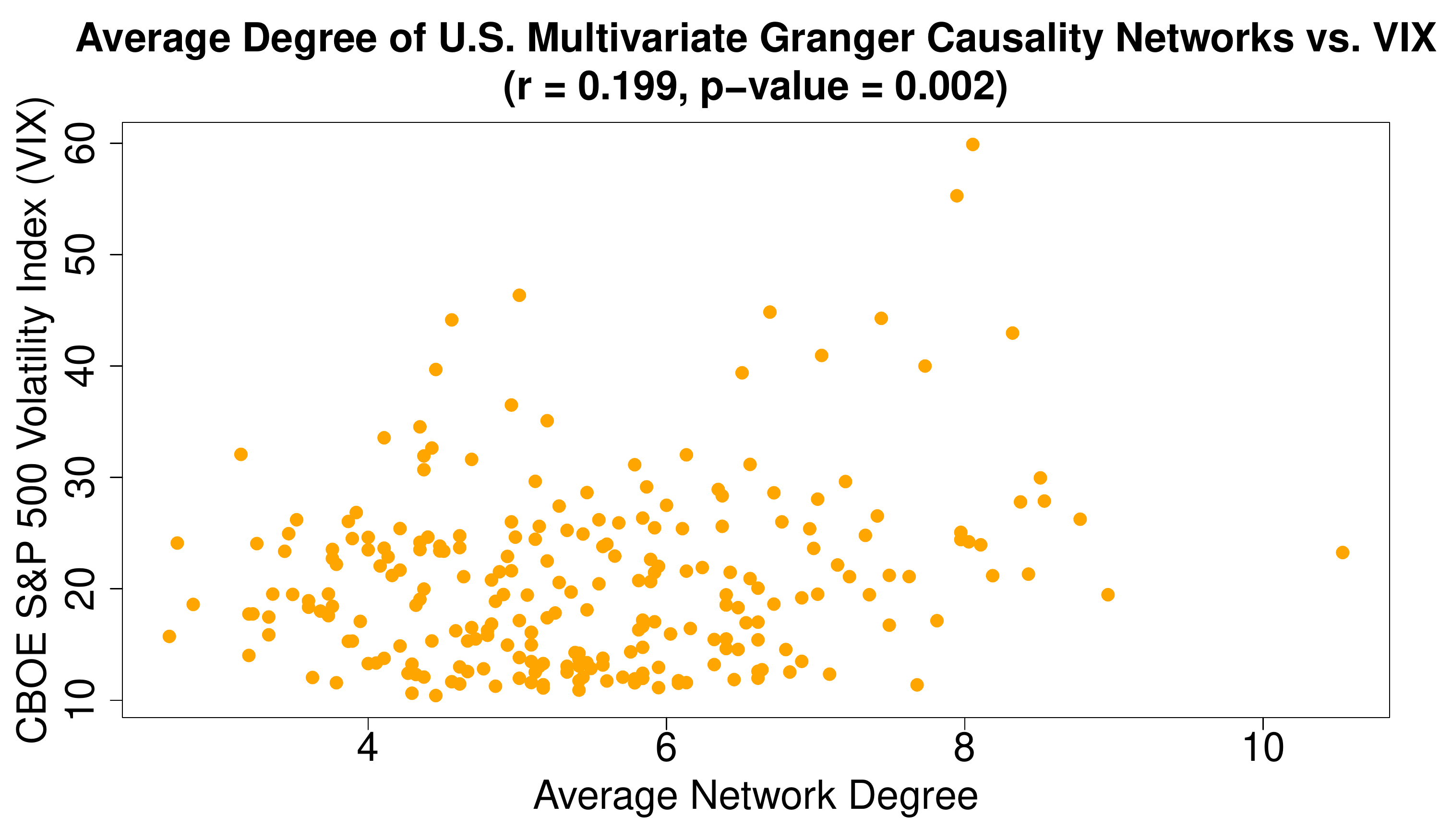}
    \end{subfigure}
    \begin{subfigure}{0.45\textwidth}
        \includegraphics[scale=0.22]{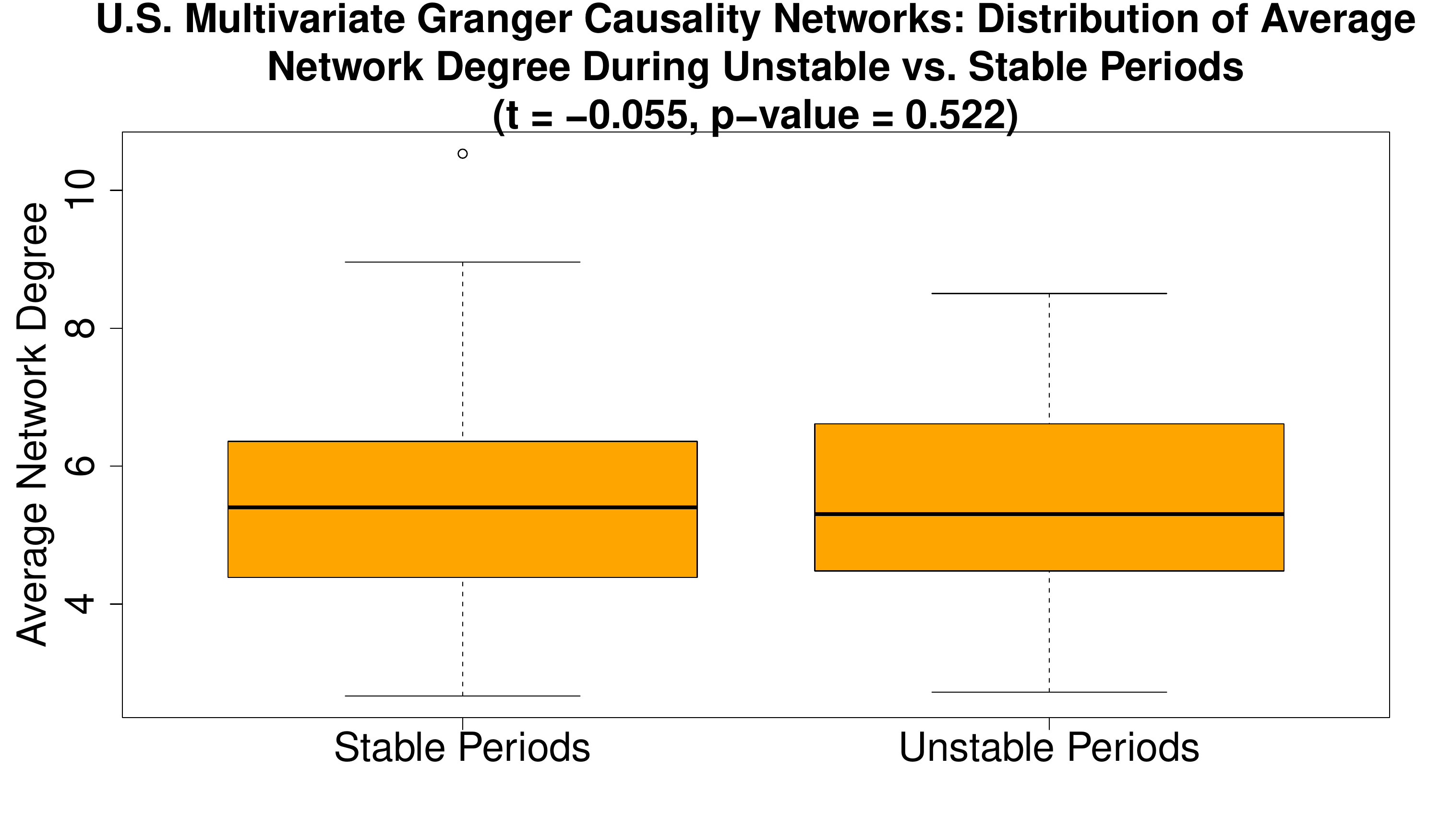}
    \end{subfigure}
     \vskip\baselineskip
    \begin{subfigure}{0.45\textwidth}
        \includegraphics[scale=0.22]{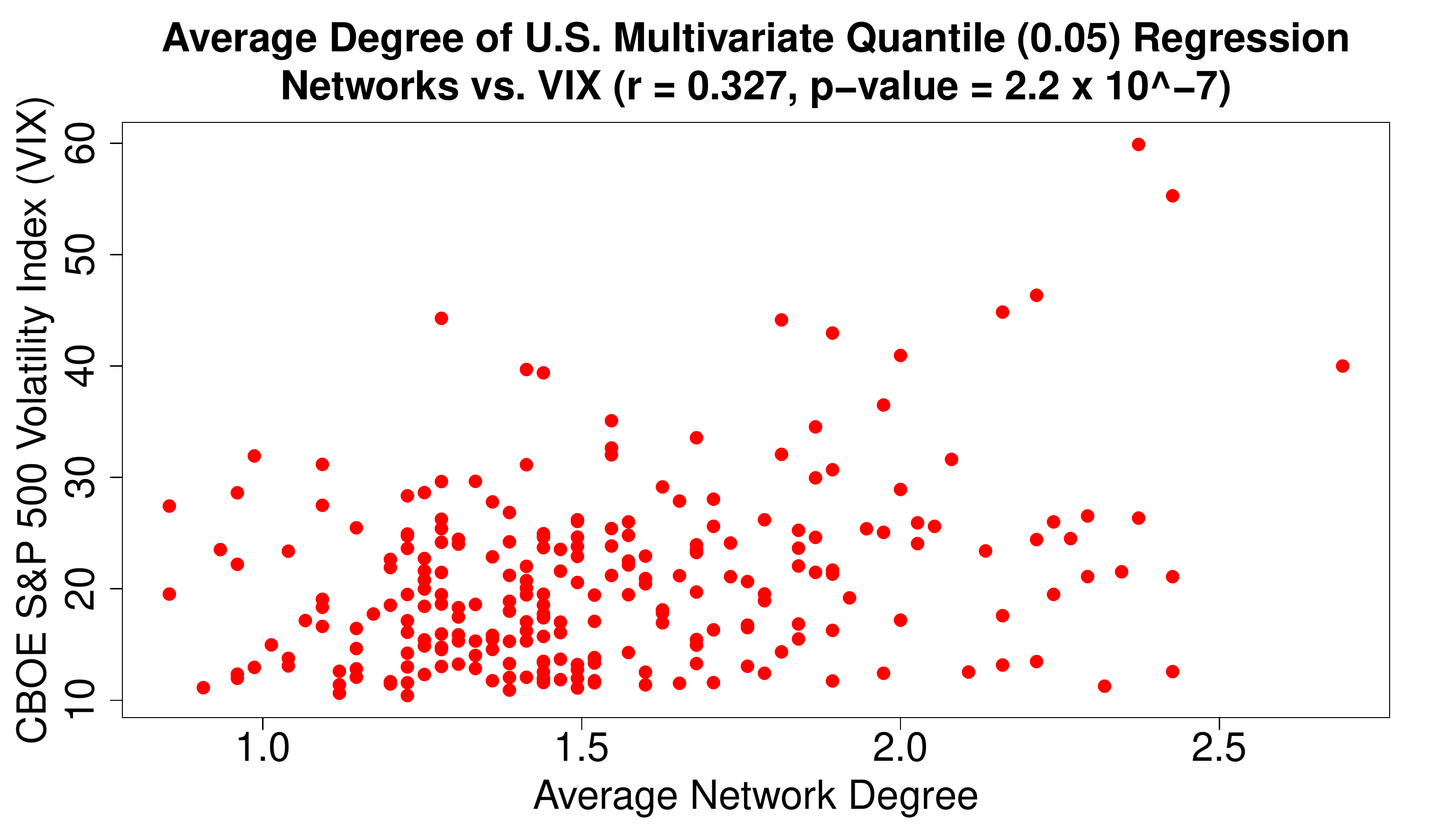}
    \end{subfigure}
    \begin{subfigure}{0.45\textwidth}
        \includegraphics[scale=0.22]{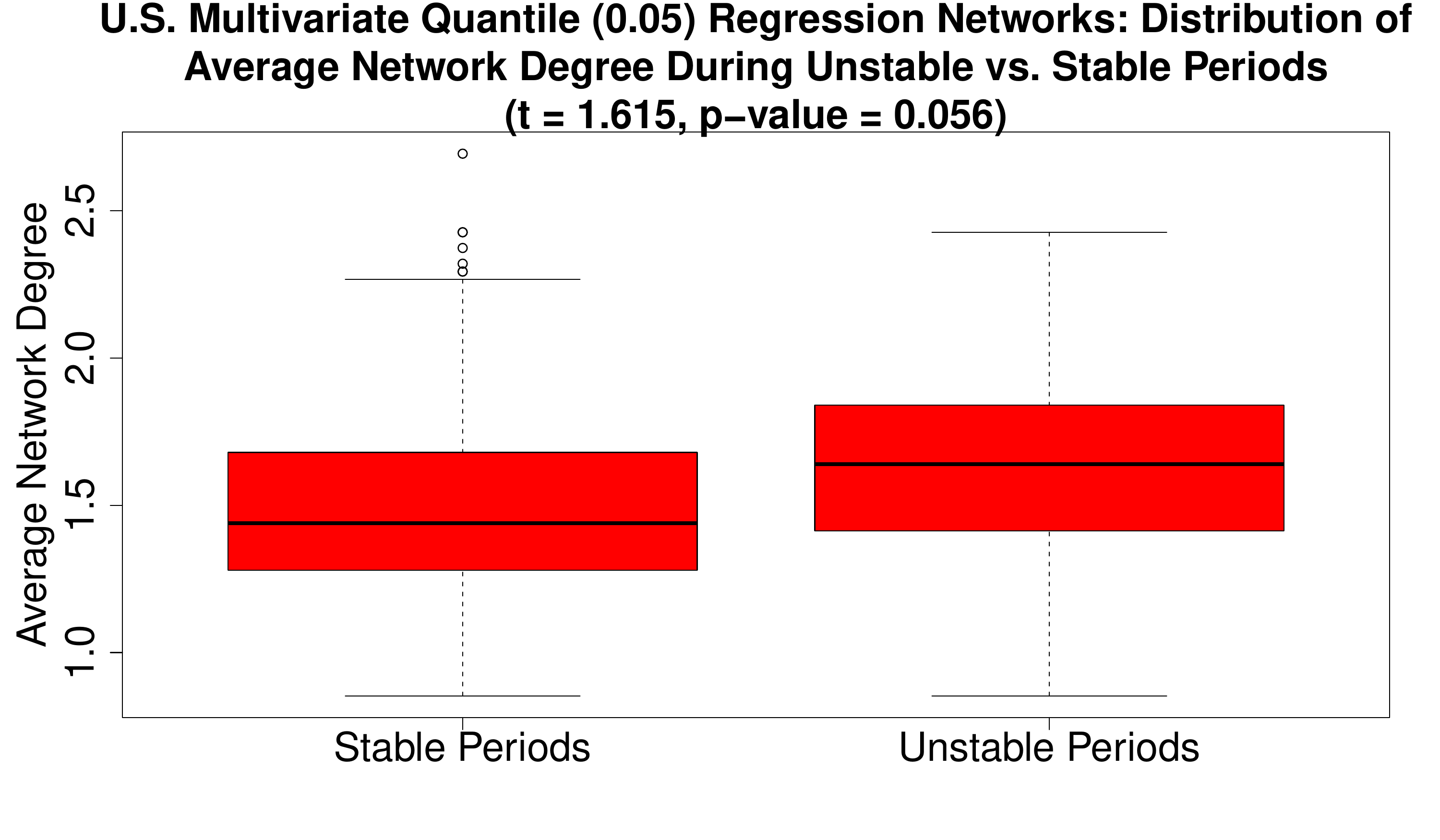}
    \end{subfigure}
    \caption{[Left]: scatterplots of VIX vs. average network degree for each 36-month window over the 1990-2012 sample period.  Correlation coefficients and their corresponding p-values are reported in the plot titles.  [Right]: boxplots displaying the distribution of average network degree during stable vs. unstable periods.  In the plot titles, we report test statistics and p-values for the t-test evaluating whether the (mean) average network degree is higher during unstable periods.  Results are for the estimated bivariate GC (top), multivariate GC (middle), and multivariate QR(0.05) (bottom) networks.
    }
    \label{fig:vix_vs_average_degree}
\end{figure}

\noindent {\textbf{Statistical Benchmarking.}} Next we perform two benchmarking analyses to evaluate whether our estimated networks become denser around systemic events.  First we measure the correlation between network degree and the Chicago Board Options Exchange's S\&P 500 Volatility Index (VIX), which is widely used to measure the expected volatility of the U.S. stock market [\citet{VIX}].  Since VIX is a proxy for investor fear, we would expect it to increase around financial crises.  In the left column of Figure \ref{fig:vix_vs_average_degree}, we plot VIX against the average network degree for each 36-month rolling window between 1990 and 2012.\footnote{The daily VIX time series is obtained from CRSP [\citet{CRSP}].  For each 36-month period, we select the value of VIX on the final trading day of that window.  We then calculate the correlation between these sampled VIX values and the average network degree.}  We display these results for the bivariate and multivariate Granger causality networks and for the multivariate QR-0.05 networks.  Correlation coefficients and their p-values are reported in the plot titles.  We observe that, for each network estimation method, average degree is significantly positively correlated with VIX, indicating that our networks tend to become more dense as investor fear increases.\footnote{We have also calculated the correlation between network degree and the S\&P 500 composite index (i.e., market) return.  Regardless of the network estimation method, we find little association between degree and market return.}  

In the right column of Figure \ref{fig:vix_vs_average_degree}, we compare the distribution of average degree in ``stable'' versus ``unstable'' periods, where, for the purposes of this analysis, we have defined unstable periods to consist of two windows preceding and following each of the systemic events marked in Figure \ref{fig:time_series_US_networks}.  The stable period consists of all other windows.  We also report (in the plot titles) the test statistic and p-value corresponding to the t-test which evaluates whether the mean network degree is higher (on average) during unstable periods.  In the case of the bivariate and multivariate Granger causality networks, we do not observe a statistically significant increase in network connectivity around systemic events (p-values are 0.881 and 0.522, respectively).  However, multivariate QR(0.05) networks \textit{do} display a statistically significant increase in average degree during the unstable periods ($t = 1.615$, p-value = 0.056).

\begin{landscape}
\begin{figure}[p]
\centering
\begin{subfigure}{0.4\textwidth}
    \centering
    \includegraphics[scale=0.33,angle=-90,trim = {0cm 3cm 2.4cm 2cm},clip]{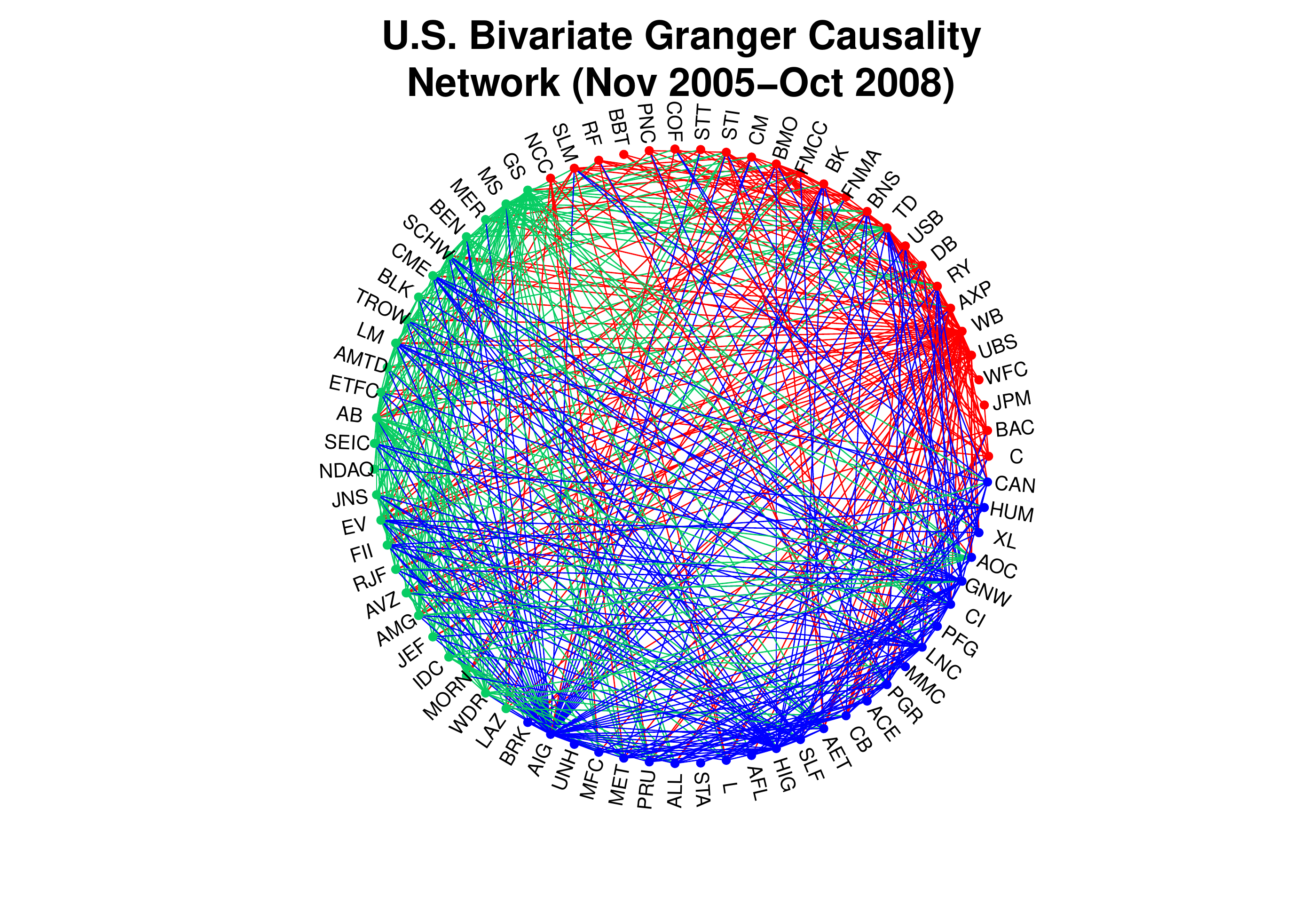}
\end{subfigure}
\begin{subfigure}{0.4\textwidth}
    \centering
    \includegraphics[scale=0.33,angle=-90,trim = {0cm 3cm 2.4cm 2cm},clip]{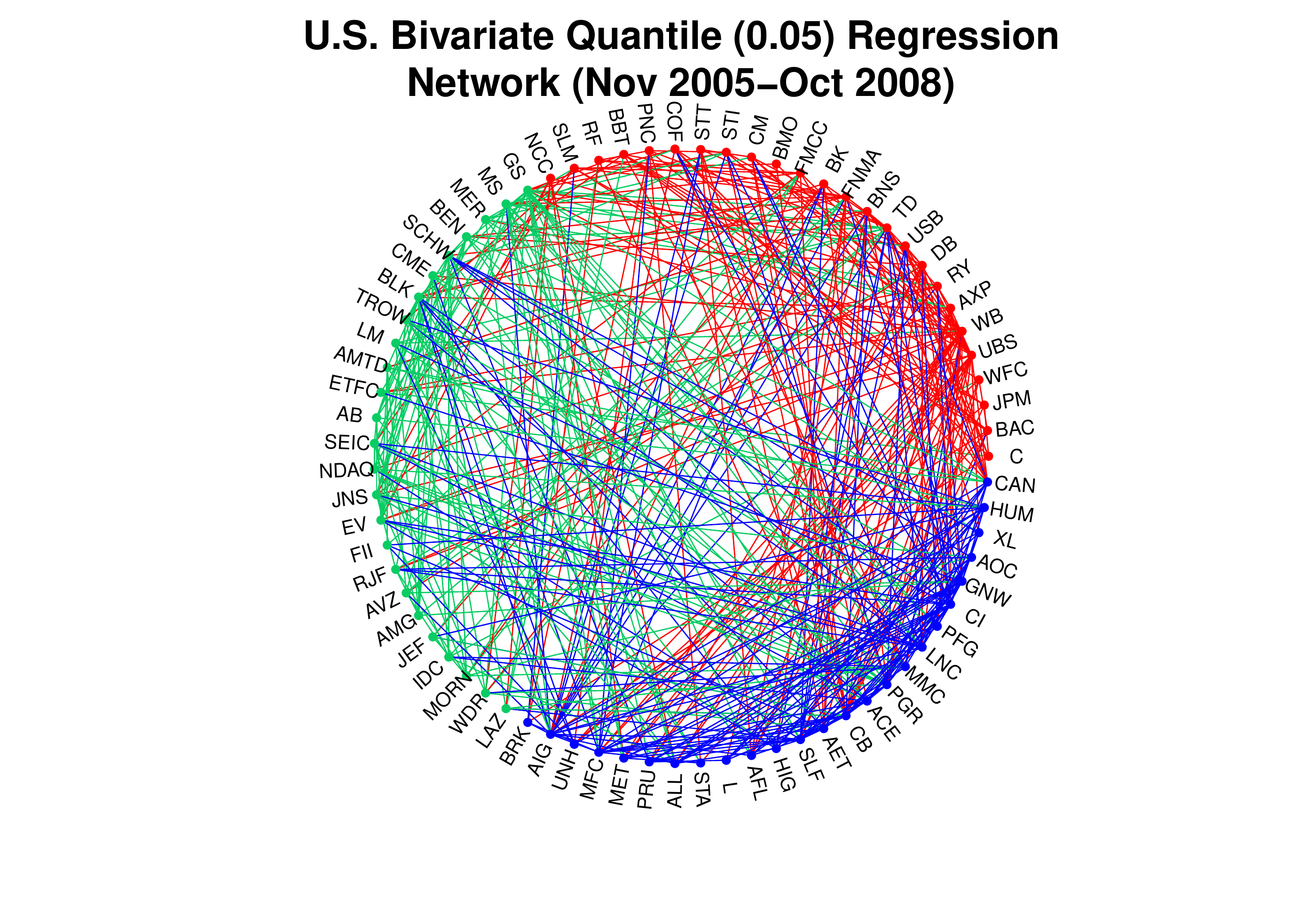}
\end{subfigure}
\begin{subfigure}{0.4\textwidth}
    \centering
    \includegraphics[scale=0.33,angle=-90,trim = {0cm 3cm 2.4cm 2cm},clip]{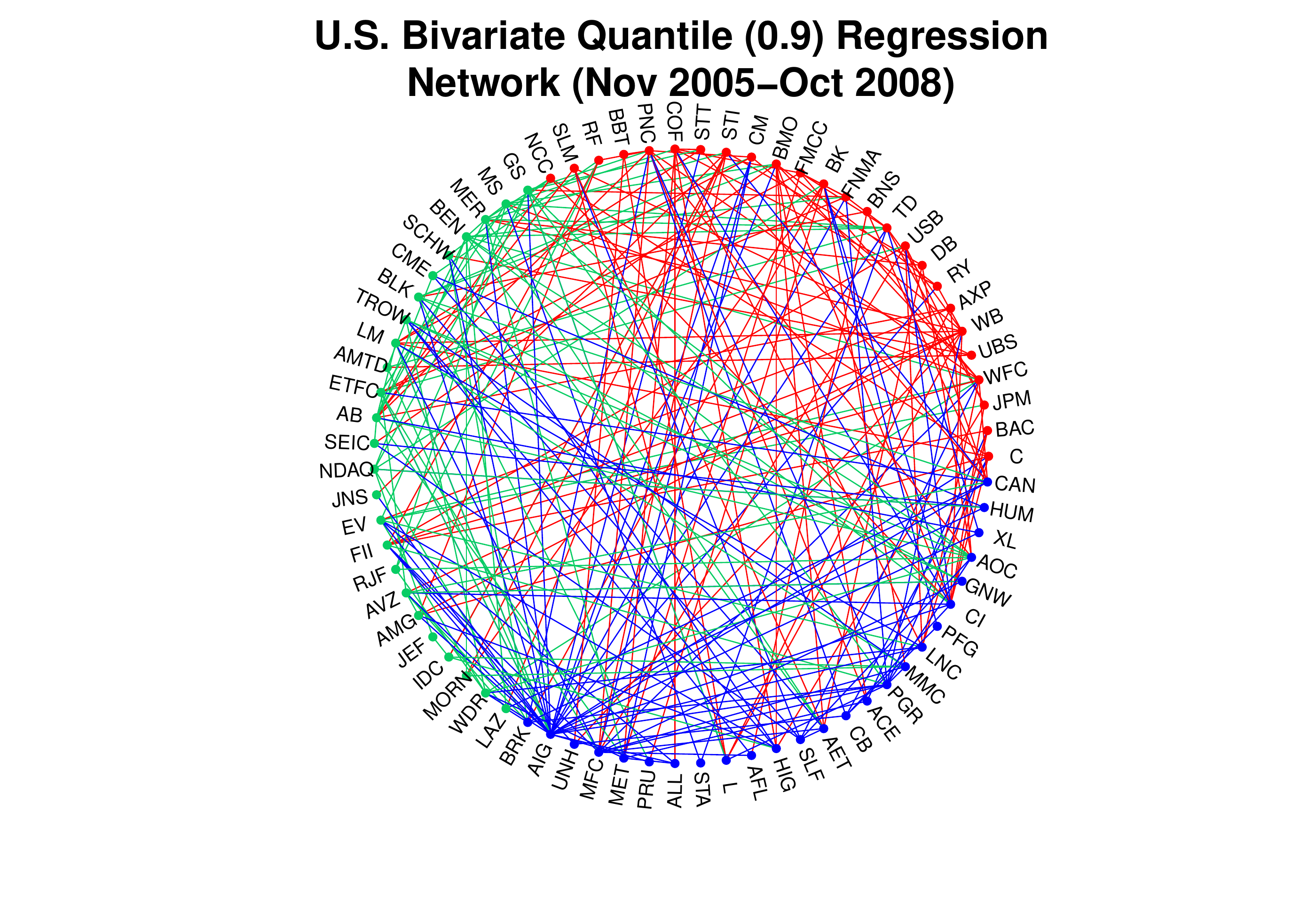}
\end{subfigure}
\vskip\baselineskip
\begin{subfigure}{0.4\textwidth}
    \centering
    \includegraphics[scale=0.33,angle=-90,trim = {0cm 3cm 2.4cm 1.8cm},clip]{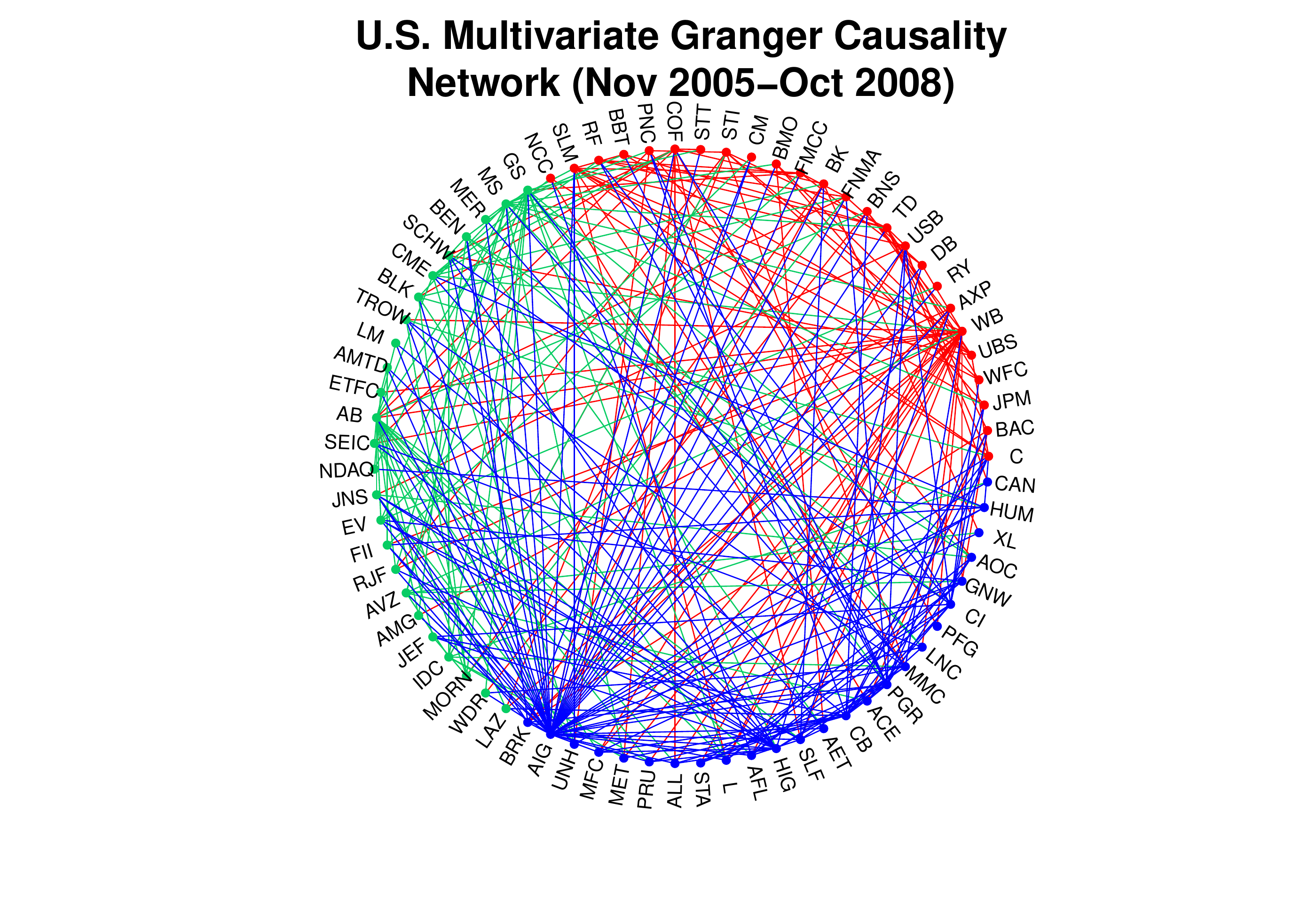}
\end{subfigure}
\begin{subfigure}{0.4\textwidth}
    \centering
    \includegraphics[scale=0.33,angle=-90,trim = {0cm 3cm 2.4cm 1.8cm},clip]{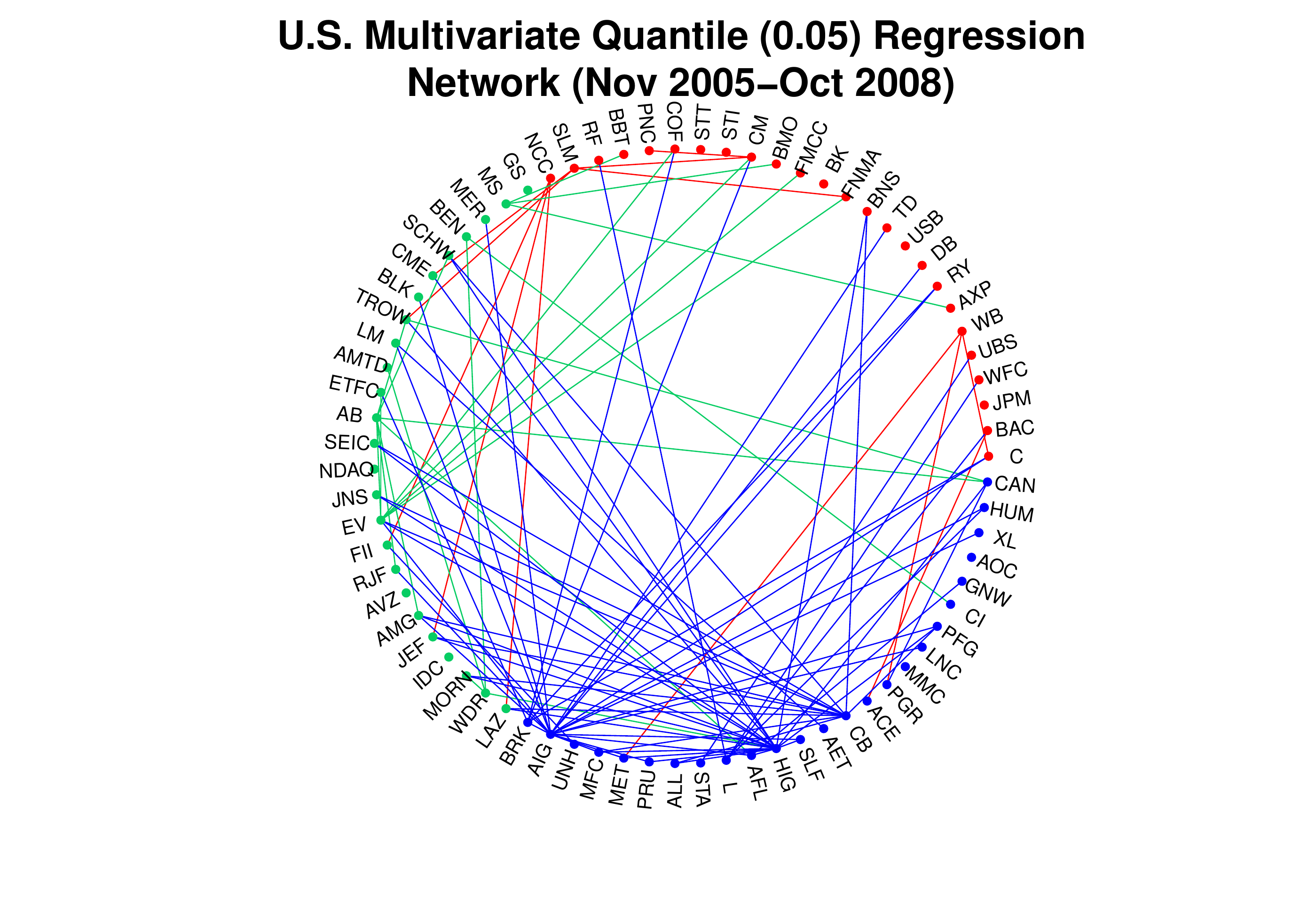}
\end{subfigure}
\begin{subfigure}{0.4\textwidth}
    \centering
    \includegraphics[scale=0.33,angle=-90,trim = {0cm 3cm 2.4cm 1.8cm},clip]{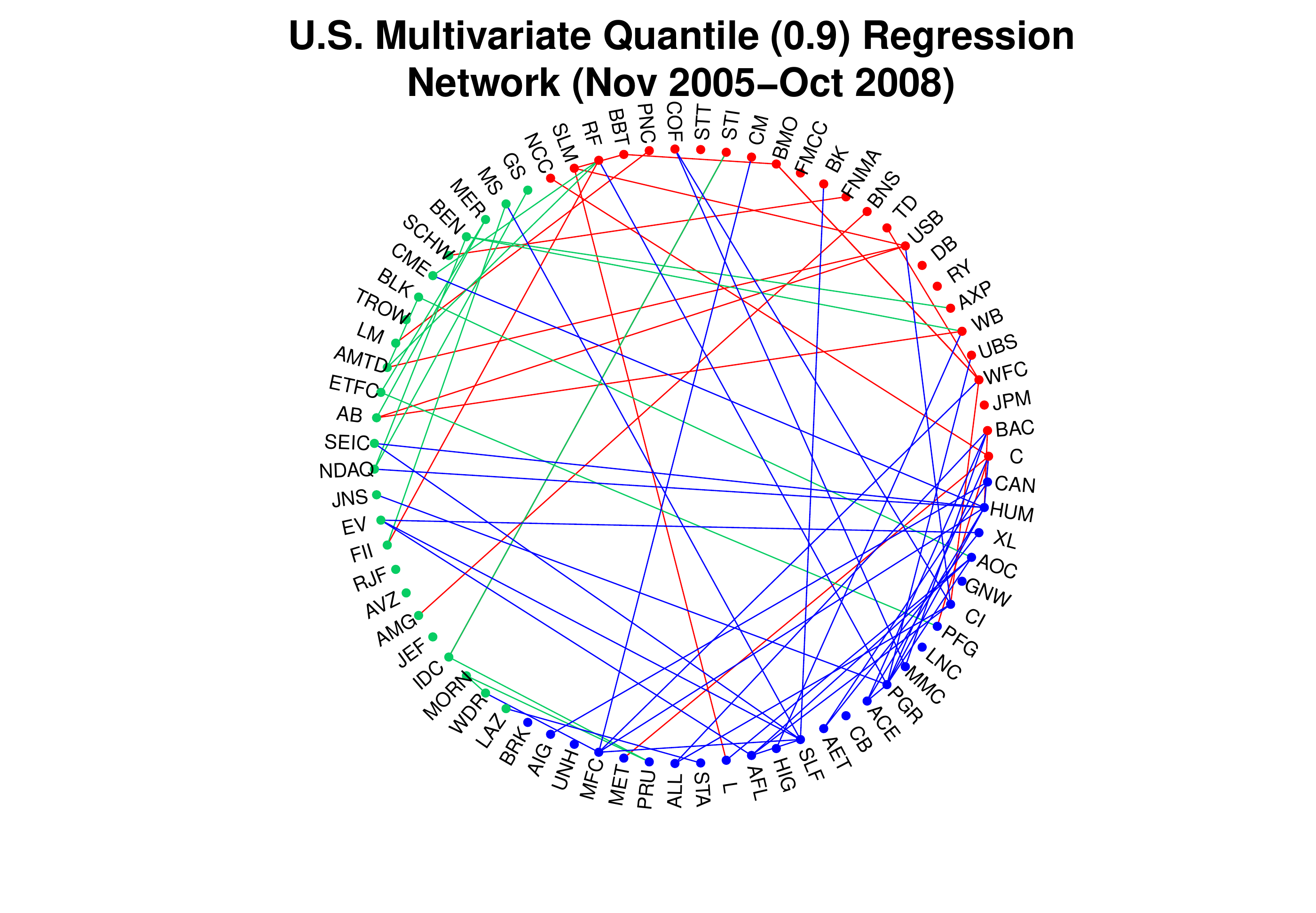}
\end{subfigure}
\caption{For November 2005-October 2008, networks estimated using GC (left), QR-0.05 (middle), and QR-0.9 (right).  Top row: bivariate results; bottom row: multivariate results.  Nodes are colored according to firm sector, with red (resp. green, blue) denoting banks (resp. broker-dealers, insurance companies).  Edges are colored according to the sector of the outgoing node.  AIG is highly connected in both the GC and QR-0.05 multivariate networks.  See Table \ref{table: QGC ticker symbol table} for the firm name corresponding to each ticker symbol.}
\label{fig:nov_05_oct_08_ntwks}
\end{figure}
\end{landscape}

\subsection{Central Firms in the U.S. Financial Sector (2007-2009)}

Having seen that our network models can perceive financial crisis periods ex-post, we now turn our attention to a different question: whether we can identify firms that play an important role in these crises.  Firms having a high degree (i.e. those that influence many other institutions) may be likely to disseminate risk through the system.  Therefore, we will concentrate our attention on the most highly connected firms in the network.

In Figure \ref{fig:nov_05_oct_08_ntwks}, we display the estimated financial networks for the November 2005 - October 2008 period, during which time the investment banking company Lehman Brothers collapsed.  Nodes are colored according to their sector: red for banks, green for broker/dealers, and blue for insurance companies.  Edge color corresponds to the sector of the node at which the edge originates; for example, if Bank of America (bank) is found to Granger-cause MetLife (insurance company), the associated edge will be colored red.  We have less power to detect relationships in the tails of the distribution, so from a statistical standpoint, we would expect the bivariate and multivariate QR-0.05 and QR-0.9 networks to have fewer edges than the corresponding GC networks.  This is the case for the multivariate networks and for bivariate QR-0.9, which has 235 edges compared to bivariate GC's 363 edges.  However, the opposite pattern holds for bivariate QR-0.05, which has 404 edges (more than GC's).  This suggests that perhaps by focusing on firms' worst returns, we can uncover connections that are not visible otherwise.  Compared to these bivariate networks, the multivariate versions are considerably less dense, as expected due to penalization.

Both bivariate and multivariate quantile regression with $\tau = 0.05$ identify many firms that are known to have played a major role in the U.S. Financial Crisis.

In Figure \ref{fig: top_firms_Nov_05_Oct_08}, we display ticker symbols for the 10 most well-connected institutions during every other 36-month rolling window from May 2007 to March 2010.  In the multivariate case, we see that American International Group (AIG) is the most well-connected firm through much of 2008-2010.  The Federal Home Loan Mortgage Corporation (Freddie Mac) and The Federal National Mortgage Association (Fannie Mae) are also highly connected in early- to mid-2008.  All 3 institutions were key drivers of the crisis, Freddie Mac and Fannie Mae through their dealings in subprime mortgages and AIG through its over-reliance on credit default swaps [\citet{causesofcrisis}].  Other highly connected firms include Goldman Sachs, Morgan Stanley, and The Hartford, all of which received major federal bailout packages during the crisis [\citet{TARP}]. 

\begin{figure}[h]
    \centering
    \includegraphics[scale=0.28,angle=-90]{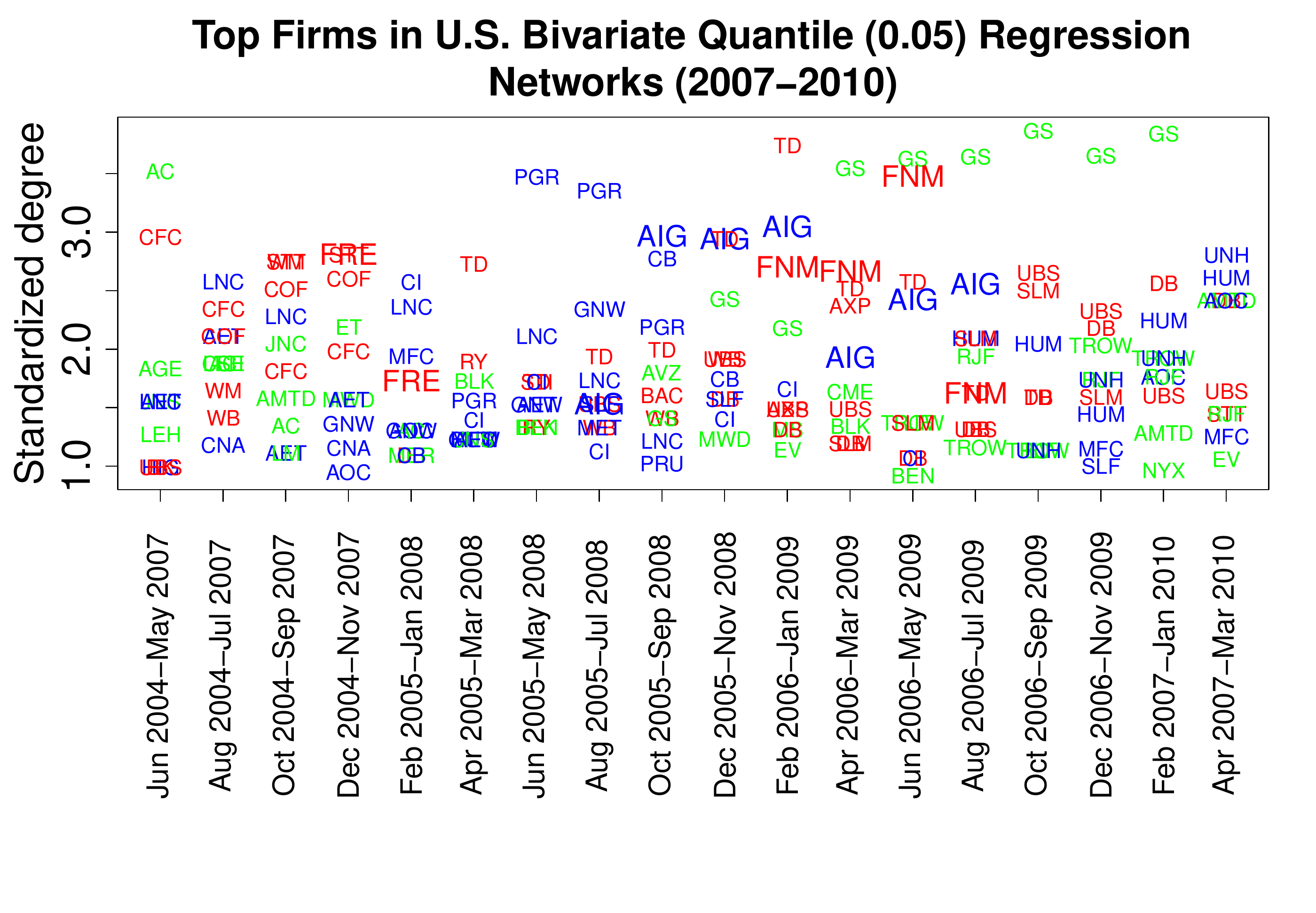}
    \includegraphics[scale=0.28,angle=-90]{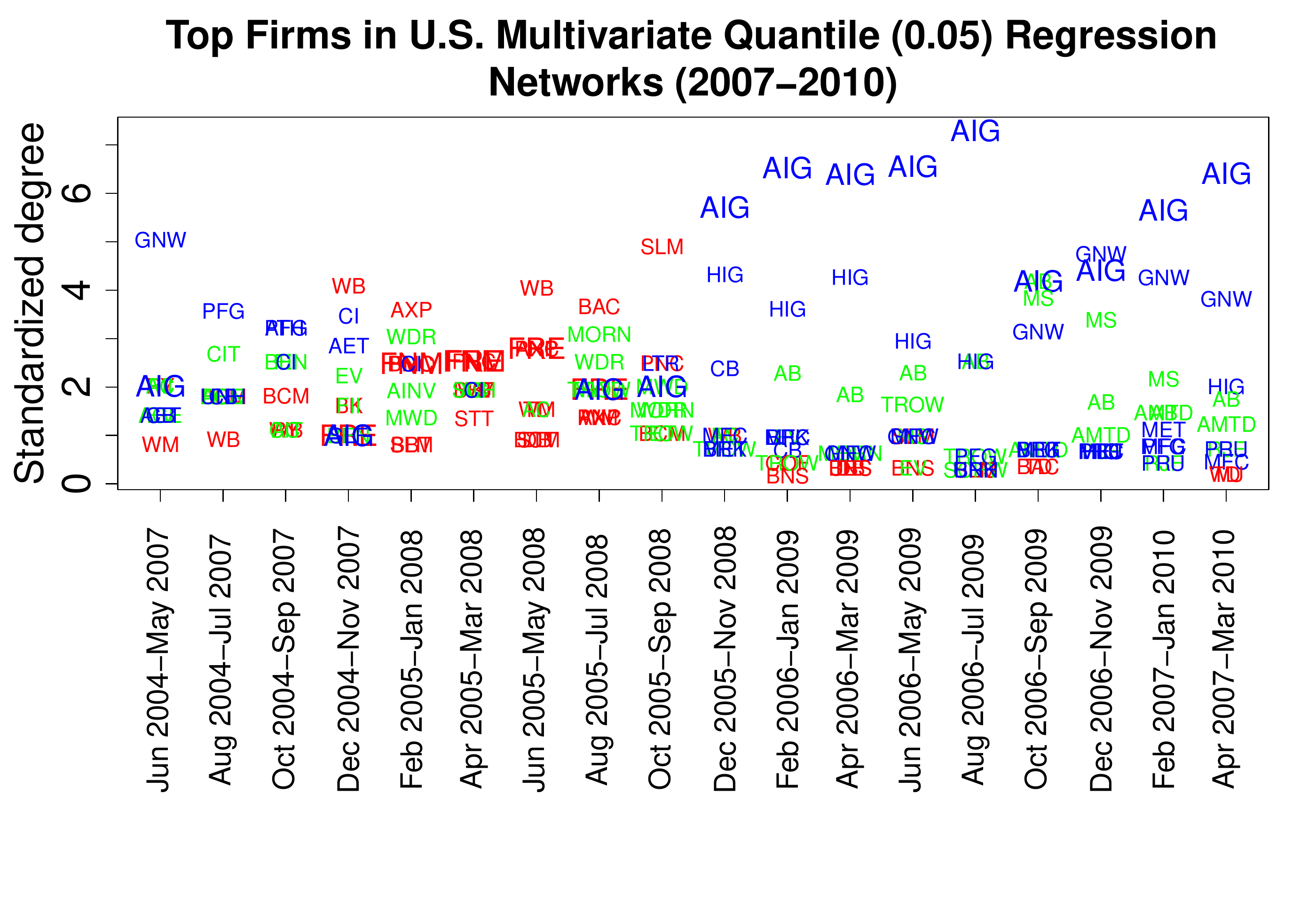}
    \caption{The top 10 firms, as ranked by their (standardized) degree, for every other window from May 2007 to March 2010.  Results from bivariate (resp. multivariate) quantile regression with $\tau = 0.05$ are presented on the left (resp. right).  Standardized degree is computed by subtracting off the average degree of all firms in the network and dividing by the standard deviation of these degrees.  Many firms that we know played a major role in the U.S. Financial Crisis appear, including Freddie Mac (FRE), Fannie Mae (FNM), and American International Group (AIG).  See Table \ref{table: QGC ticker symbol table} for the firm name corresponding to each ticker symbol.}
    \label{fig: top_firms_Nov_05_Oct_08}
\end{figure}

\begin{table}[h]
\centering
\setlength{\tabcolsep}{7pt}
\begin{tabular}{c|cccccc}
\toprule
& \multicolumn{3}{c}{\textbf{BIVARIATE}} & \multicolumn{3}{c}{\textbf{MULTIVARIATE}} \\ 
\textbf{Rank} & \textbf{GC} & \textbf{QR 0.05} & \textbf{QR 0.9} & \textbf{GC} & \textbf{QR 0.05} & \textbf{QR 0.9} 
\\ \midrule
1 & AIG (58) & AIG (26) & AIG (29) & AIG (47) & AIG (20) & HUM (8) \\
2 & WB (45) & CB (26) & CI (14) & WB (24) & HIG (20) & PGR (7) \\
3 & AC (36) & GS (25) & MFC (13) & AC (22) & CB (12) & SLF (6) \\
4 & MS (36) & UBS (25) & AC (12) & PGR (16) & EV (9) & CI (5) \\
5 & HIG (32) & SLF (24) & PNC (12) & CB (16) & AC (7) & MFC (5) \\
6 & TD (31) & TD (24) & WB (12) & HIG (16) & BRK (6) & C (5) \\
7 & GS (27) & CI (23) & WFC (12) & MS (16) & CNA (4) & AON (4) \\
8 & GNW (26) & WB (23) & PGR (12) & GS (16) & AFL (4) & AFL (4) \\
9 & LNC (25) & PGR (23) & TD (11) & CI (13) & L (4) & RF (4) \\
10 & EV (24) & GNW (22) & AON (11) & SLM (12) & WDR (4) & USB (4) \\
\bottomrule
\end{tabular}
\caption{\label{table: top_firms_Nov_05_Oct_08} The top 10 firms, as ranked by their degree, for November 2005-October 2008.  The left 3 columns present bivariate results and the right 3, their multivariate analogues.  Observe that both bivariate and multivariate mean-based and lower tail analyses highlight the role of American International Group (AIG), which is known to have played a pivotal role in the U.S. Financial Crisis.  See Table \ref{table: QGC ticker symbol table} for the firm name corresponding to each ticker symbol.}
\end{table}

\section{Empirical Results: Indian Banks}
\label{sec:India_empirical_section}

Having estimated financial networks based on historical U.S. data, we now turn our attention to performing a similar analysis on the stock returns of Indian banks.  In Section \ref{sec:US_empirical_section}, we focused on how lower-tail analysis can detect systemic risk.  In this section, on the other hand, we explore how different economic events are highlighted by different quantiles; that is, we observe that lower-tail networks tend to display increased connectivity when the markets receive potentially ``bad news'' while upper-tail networks exhibit the same around announcements of potentially ``good news.'' \\

\smallskip
\noindent {\textbf{Data Collection and Pre-processing. }}
We select the top 30 Indian banks, as defined by the Thomson Reuters Industry Classification Benchmark, again measuring firm size by market capitalization. We then use the banks' monthly stock returns to estimate a network for each 36-month rolling window from December 2000 through January 2018. Networks are estimated using both multivariate GC and QGC with $\tau = 0.2$ (lower tail analysis) and $\tau = 0.8$ (upper tail analysis).
 
Lastly, as in Section \ref{sec:US_empirical_section}, we measure connectivity via the average degree of the undirected networks.

\subsection{Evolution of Network Summary Statistics}
In Figure \ref{fig:india_bank_time_series_avg_degree}, we plot the average degree of the estimated India bank networks over the entire sample period.  These time series illustrate how quantile-based analysis can make the results of Granger causality analysis more interpretable.  For example, in networks estimated using Granger causality, we see a large and persistent increase in average degree from early 2006 through mid-2008; in fact, GC connectivity does not reach these levels during any other portion of the sample period.  When we perform upper and lower tail analyses (focusing on good and bad days, respectively), we see that the QR 0.8 networks also display increased connectivity during 2006-2008, while the pattern is mixed for the QR 0.2 networks.  (These have high average degree throughout 2006, but then connectivity drops and we do not see a persistent increase in degree thereafter.)  Thus it appears that most of the connections in the 2006-2008 Granger causality networks are driven by what happens on good days in the market.

Similarly, the average degree spikes in both Granger causality and QR 0.8 networks in late 2012 and in September 2013 (see Figure \ref{fig:interesting_ntwks_India}).  The latter is particularly interesting since Narendra Modi was named as the prime ministerial candidate of the Bharatiya Janata party (BJP) on September 13, 2013 [\citet{modiElection}].  Upper tail connectivity remains high for several months after the announcement, perhaps reflecting shared confidence in the economy.  (Noteably lower tail networks have low average degree during this period.)  The average degrees of GC and QR 0.8 networks also increase in early 2015 and remain elevated through mid-2016.  Lastly, starting in early 2016 and continuing through April 2017, the QR 0.2 networks display a large increase in average degree, possibly due to rumors of and later actual banknote demonetisation, which occurred in November 2016 (see Figure \ref{fig:interesting_ntwks_India}) [\citet{demonetisation}].  Importantly, neither of these connectivity spikes is clearly captured by Granger causality or upper tail networks, which experience only a small connectivity increase \textit{after} demonetisation.  This demonstrates that (a) upper and lower quantile analysis can reveal economic patterns not highlighted by Granger causality, and (b) by varying the quantile under consideration, we can detect different types of economic events.

\begin{figure}[h]
    \centering
    \includegraphics[scale=0.48,angle=-90]{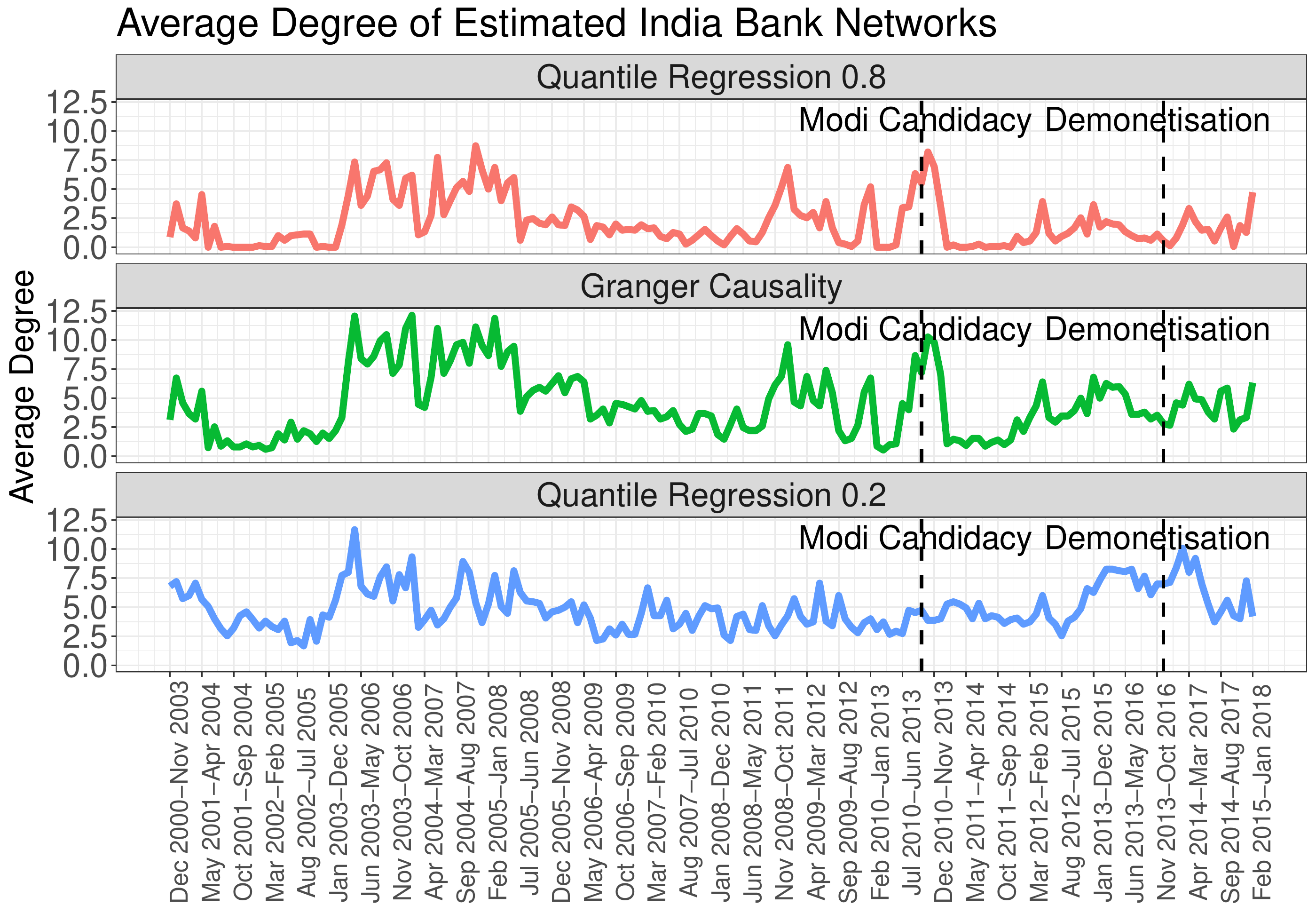}
    \caption{Average degree of undirected networks for the top 30 Indian banks.  Networks were estimated using Granger causality (middle) and quantile regression with $\tau = 0.8$ (top) and $\tau = 0.2$ (bottom).  The announcement of Narendra Modi as a prime ministerial candidate (September 2013) and Indian banknote demonetisation (November 2016) are marked with dashed lines.}
    \label{fig:india_bank_time_series_avg_degree}
\end{figure}

\begin{figure}[h]
    \centering
    \includegraphics[scale=0.33,angle=-90,trim = {0cm 4cm 2cm 2cm}, clip]{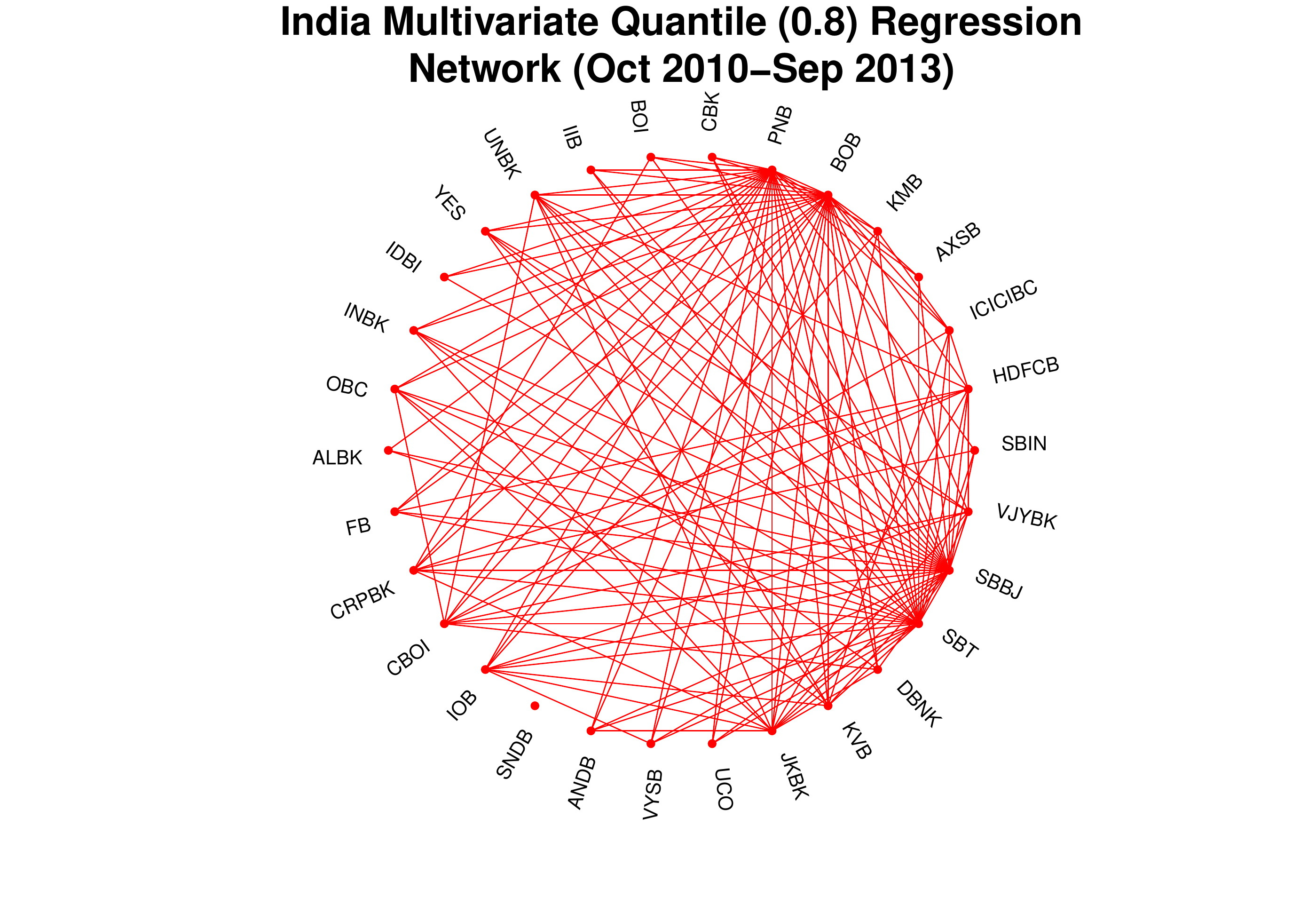}
    \includegraphics[scale=0.33,angle=-90,trim = {0cm 3cm 2cm 1cm}, clip]{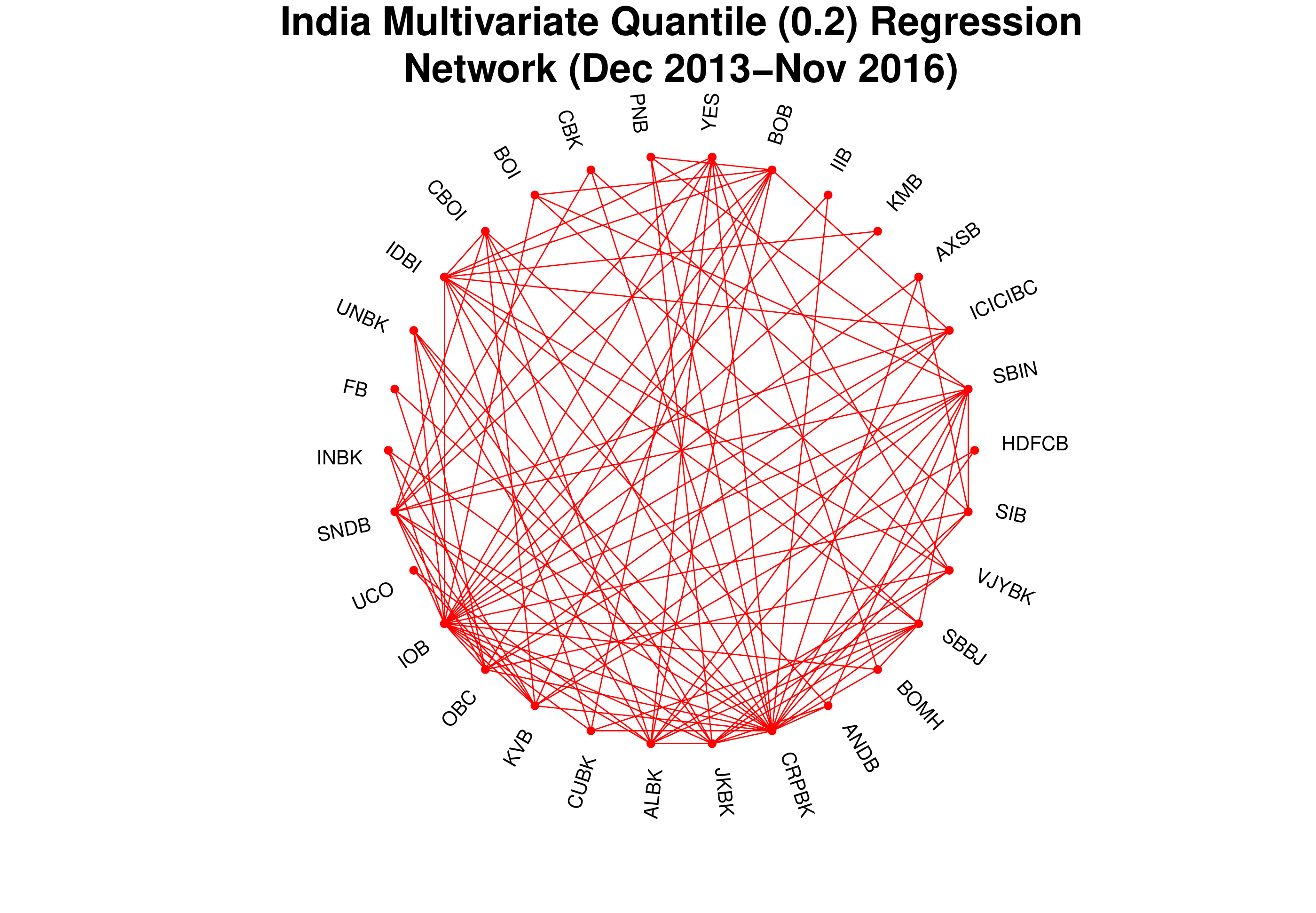}
    \caption{Estimated India bank networks for October 2010-September 2013 (left) and December 2013-November 2016 (right).  The former time period culminated in the announcement of Narendra Modi as a prime ministerial candidate and the latter time period culminated in Indian banknote demonetisation.  See Table \ref{table: QGC ticker symbol table for India} for the full firm name corresponding to each ticker symbol.}
    \label{fig:interesting_ntwks_India}
\end{figure}
\section{Conclusion}
Systemic risk is a complex economic concept for which no one metric is sufficient.  Multiple risk measures are needed to provide financial regulators with the tools required to implement sound policies.  We have proposed one such measure that yields quantile-based networks.  In contrast to pairwise approaches, our  method conditions on all firms in the sample to avoid the appearance of spurious links.  We apply our risk measure to historical returns of large U.S. and Indian firms and study the connectivity of the resulting networks.  We find that our lower tail method detects financial crises and identifies systemically-important firms, like AIG and Fannie Mae for the US, that are not highlighted at other quantiles.  

Future work in this direction includes developing the theory of lasso penalized quantile regression on multivariate time series data to enable easier tuning parameter selection.  In particular, it will be useful to have uncertainty measures associated with each edge estimate.  Lastly, we may be able to reduce the dimensionality of our problem by exploiting sectoral information; for instance, we could apply group lasso, where each group corresponds to a different sector or sub-sector, in order to reduce the number of parameters we must estimate.

\section*{Acknowledgement}

SB gratefully acknowledges support from NSF awards DMS-1812128, DMS-2210675, and NIH awards  R01GM135926 and R21NS120227.  We also thank David Easley, Stephen P. Ellner, and Steven Strogatz for their helpful comments.

\bibliographystyle{abbrvnat}

\bibliography{dsmmbib}

\begin{thebibliography}{27}
\providecommand{\natexlab}[1]{#1}
\providecommand{\url}[1]{\texttt{#1}}
\expandafter\ifx\csname urlstyle\endcsname\relax
  \providecommand{\doi}[1]{doi: #1}\else
  \providecommand{\doi}{doi: \begingroup \urlstyle{rm}\Url}\fi

\bibitem[VIX()]{VIX}
Cboe vix.
\newblock \url{http://www.cboe.com/vix}.
\newblock Accessed: 2020-08-03.

\bibitem[dem()]{demonetisation}
India scraps 500 and 1,000 rupee bank notes overnight.
\newblock \url{https://www.bbc.com/news/business-37906742}.
\newblock Accessed: 2020-07-02.

\bibitem[mod()]{modiElection}
India polls: Narendra modi revealed as bjp's pm candidate.
\newblock \url{https://www.bbc.com/news/world-middle-east-24080193}.
\newblock Accessed: 2020-07-02.

\bibitem[Ahelegbey et~al.(2016)Ahelegbey, Billio, and
  Casarin]{ahelegbey2016sparse}
D.~F. Ahelegbey, M.~Billio, and R.~Casarin.
\newblock Sparse graphical vector autoregression: a bayesian approach.
\newblock \emph{Annals of Economics and Statistics/Annales d'{\'E}conomie et de
  Statistique}, \penalty0 (123/124):\penalty0 333--361, 2016.

\bibitem[Barrodale and Roberts(1974)]{barrodalerobertsalgorithm}
I.~Barrodale and F.~Roberts.
\newblock Solution of an overdetermined system of equations in the $\ell_1$
  norm.
\newblock \emph{Communications of the ACM}, 17\penalty0 (6):\penalty0 319--320,
  1974.

\bibitem[Basu et~al.(2019)Basu, Das, Michailidis, and Purnanandam]{preprint}
S.~Basu, S.~Das, G.~Michailidis, and A.~K. Purnanandam.
\newblock A system-wide approach to measure connectivity in the financial
  sector.
\newblock \emph{Available at SSRN 2816137}, 2019.

\bibitem[Billio et~al.(2012)Billio, Getmansky, Lo, and
  Pelizzon]{billio2012econometric}
M.~Billio, M.~Getmansky, A.~W. Lo, and L.~Pelizzon.
\newblock Econometric measures of connectedness and systemic risk in the
  finance and insurance sectors.
\newblock \emph{Journal of Financial Economics}, 104\penalty0 (3):\penalty0
  535--559, 2012.

\bibitem[Bollerslev(1986)]{GARCH}
T.~Bollerslev.
\newblock Generalized autoregressive conditional heteroskedasticity.
\newblock \emph{Journal of Econometrics}, 31\penalty0 (3):\penalty0 307--327,
  1986.

\bibitem[{CRSP Stocks}()]{CRSP}
{CRSP Stocks}.
\newblock {Available: Center for Research in Security Prices. Graduate School
  of Business. University of Chicago. Retrieved from Wharton Research Data
  Services}.
\newblock Accessed 2018.

\bibitem[Demirer et~al.(2018)Demirer, Diebold, Liu, and
  Yilmaz]{demirer2018estimating}
M.~Demirer, F.~X. Diebold, L.~Liu, and K.~Yilmaz.
\newblock Estimating global bank network connectedness.
\newblock \emph{Journal of Applied Econometrics}, 33\penalty0 (1):\penalty0
  1--15, 2018.

\bibitem[Diebold and Y{\i}lmaz(2014)]{diebold2014network}
F.~X. Diebold and K.~Y{\i}lmaz.
\newblock On the network topology of variance decompositions: Measuring the
  connectedness of financial firms.
\newblock \emph{Journal of Econometrics}, 182\penalty0 (1):\penalty0 119--134,
  2014.

\bibitem[{Financial Crisis Inquiry Commission}(2011)]{causesofcrisis}
{Financial Crisis Inquiry Commission}.
\newblock {The Financial Crisis Inquiry Report. Final Report of the National
  Commission on the Causes of the Financial and Economic Crisis in the United
  States}, 2011.

\bibitem[Fu and Knight(2000)]{fu2000asymptotics}
W.~Fu and K.~Knight.
\newblock Asymptotics for lasso-type estimators.
\newblock \emph{The Annals of statistics}, 28\penalty0 (5):\penalty0
  1356--1378, 2000.

\bibitem[Granger(1969)]{Granger}
C.~W. Granger.
\newblock Investigating causal relations by econometric models and
  cross-spectral methods.
\newblock \emph{Econometrica}, 37\penalty0 (3):\penalty0 424--438, 1969.

\bibitem[Hall and Heyde(1980)]{hall1980martingale}
P.~Hall and C.~C. Heyde.
\newblock \emph{Martingale limit theory and its application}.
\newblock Academic press, 1980.

\bibitem[Hamilton and Press(1994)]{hamilton1994time}
J.~Hamilton and P.~U. Press.
\newblock \emph{Time Series Analysis}.
\newblock Number v. 10 in Book collections on Project MUSE. Princeton
  University Press, 1994.
\newblock ISBN 9780691042893.
\newblock URL \url{https://books.google.com/books?id=B8\_1UBmqVUoC}.

\bibitem[H{\"a}rdle et~al.(2016)H{\"a}rdle, Wang, and Yu]{hardle2016tenet}
W.~K. H{\"a}rdle, W.~Wang, and L.~Yu.
\newblock Tenet: Tail-event driven network risk.
\newblock \emph{Journal of Econometrics}, 192\penalty0 (2):\penalty0 499--513,
  2016.

\bibitem[Herce(1996)]{herce}
M.~A. Herce.
\newblock Asymptotic theory of ``lad'' estimation in a unit root process with
  finite variance errors.
\newblock \emph{Econometric Theory}, 12, 1996.

\bibitem[Koenker(2005)]{quantileregressionKoenker}
R.~Koenker.
\newblock \emph{Quantile Regression}.
\newblock Cambridge University Press, 2005.

\bibitem[Koenker and Xiao(2006)]{koenker_xiao}
R.~Koenker and Z.~Xiao.
\newblock Quantile autoregression.
\newblock \emph{Journal of the American Statistical Association}, 101\penalty0
  (475):\penalty0 980--990, 2006.

\bibitem[Koenker et~al.(2014)Koenker, Mizera, et~al.]{koenker2014convex}
R.~Koenker, I.~Mizera, et~al.
\newblock Convex optimization in {R}.
\newblock \emph{Journal of Statistical Software}, 60\penalty0 (5):\penalty0
  1--23, 2014.

\bibitem[L{\"u}tkepohl(2005)]{lutkepohl}
H.~L{\"u}tkepohl.
\newblock \emph{New introduction to multiple time series analysis}.
\newblock Springer Science \& Business Media, 2005.

\bibitem[Pollard(1991)]{pollard1991asymptotics}
D.~Pollard.
\newblock Asymptotics for least absolute deviation regression estimators.
\newblock \emph{Econometric Theory}, 7\penalty0 (2):\penalty0 186--199, 1991.

\bibitem[{ProPublica}()]{TARP}
{ProPublica}.
\newblock {Bailout Recipients}.
\newblock \emph{Available at https://projects.propublica.org/bailout/list}.
\newblock {Accessed 2018}.

\bibitem[Tibshirani(1996)]{Lasso}
R.~Tibshirani.
\newblock Regression shrinkage and selection via the lasso.
\newblock \emph{Journal of the Royal Statistical Society. Series B
  (Methodological)}, 58\penalty0 (1):\penalty0 267--288, 1996.

\bibitem[Tobias and Brunnermeier(2016)]{adrian2011covar}
A.~Tobias and M.~K. Brunnermeier.
\newblock Covar.
\newblock \emph{The American Economic Review}, 106\penalty0 (7):\penalty0 1705,
  2016.

\bibitem[Wu and Liu(2009)]{wu2009variable}
Y.~Wu and Y.~Liu.
\newblock Variable selection in quantile regression.
\newblock \emph{Statistica Sinica}, pages 801--817, 2009.

\end{thebibliography}

\newpage
\begin{center}
{\large\bf SUPPLEMENTARY MATERIAL}
\end{center}

\appendix
\section{List of U.S. Financial Firms}

\begin{longtable}[H]{c | c | c}
\textbf{Firm Name} & \textbf{Sector} & \textbf{Ticker Symbol} \\ \hline\hline
A F L A C INC & INS & AFL \\
ACE LTD & INS & ACE \\
AETNA INC NEW & INS & AET \\
AFFILIATED MANAGERS GROUP INC & PB & AMG \\
ALLIANCE CAPITAL MGMT HLDG L P & PB & AC \\
ALLSTATE CORP & INS & ALL  \\
AMERICAN EXPRESS CO & BA & AXP \\
AMERICAN INTERNATIONAL GROUP & INS & AIG \\
APOLLO INVESTMENT CORP & PB & AINV \\
ASSURANT INC & INS & AIZ \\
AMERITRADE HOLDING CORP NEW & PB & AMTD \\
AMVESCAP PLC & PB & AVZ \\
AON CORP & INS & AON \\
B B \& T CORP & BA & BBT \\
BANK MONTREAL QUE & BA & BMO \\
BANK NEW YORK INC & BA & BK \\
BANK OF AMERICA CORP & BA & BAC \\
BANK OF NOVA SCOTIA & BA & BNS \\
BERKSHIRE HATHAWAY INC DEL & INS & BRK \\
BLACKROCK INC & PB & BLK \\
C I G N A CORP & INS & CI \\
C N A FINANCIAL CORP & INS & CNA \\
CANADIAN IMPERIAL BANK COMMERCE & BA & CM \\
CAPITAL ONE FINANCIAL CORP & BA & COF \\
CHICAGO MERCANTILE EXCH HLDG INC & PB & CME \\
CHUBB CORP & INS & CB \\
CITIGROUP & BA & C \\
DEUTSCHE BANK A G & BA & DB \\
E TRADE FINANCIAL CORP & PB & ETFC \\
EATON VANCE CORP & PB & EV \\
FEDERAL HOME LOAN MORTGAGE CORP & BA & FRE \\
FEDERAL NATIONAL MORTGAGE ASSN & BA & FNM \\
FEDERATED INVESTORS INC PA & PB & FII \\
FRANKLIN RESOURCES INC & PB & BEN \\
GENWORTH FINANCIAL INC & INS & GNW \\
GOLDMAN SACHS GROUP INC & PB & GS \\
HARTFORD FINANCIAL SVCS GRP INC & PB & HIG \\
HUMANA INC & INS & HUM \\
INTERACTIVE DATA CORP & PB & IDC \\
JANUS CAP GROUP INC & PB & JNS \\
JEFFERIES GROUP INC NEW & PB & JEF \\
JPMORGAN CHASE \& CO & BA & JPM \\
LAZARD LTD & PB & LAZ \\
LEGG MASON INC & PB & LM \\
LINCOLN NATIONAL CORP IN & INS & LNC \\
LOEWS CORP & INS & L \\
MANULIFE FINANCIAL CORP & INS & MFC \\
MARSH \& MCLENNAN COS INC & INS & MMC \\
MERRILL LYNCH \& CO INC & PB & MER \\
METLIFE INC & INS & MET \\
MORGAN STANLEY DEAN WITTER \& CO & PB & MS \\
MORNINGSTAR INC & PB & MORN \\
NASDAQ STOCK MARKET INC & PB & NDAQ \\
NATIONAL CITY CORP & BA & NCC \\
P N C FINANCIAL SERVICES GRP INC & BA & PNC \\
PRINCIPAL FINANCIAL GROUP INC & INS & PFG \\
PROGRESSIVE CORP OH & INS & PGR \\
PRUDENTIAL FINANCIAL INC & INS & PRU \\
RAYMOND JAMES FINANCIAL INC & PB & RJF \\
REGIONS FINANCIAL CORP & BA & RF \\
ROYAL BANK CANADA MONTREAL QUE & BA & RY \\
S E I INVESTMENTS COMPANY & PB & SEIC \\
SLM CORP & BA & SLM \\
SCHWAB CHARLES CORP NEW & PB & SCHW \\
ST PAUL TRAVELERS COS INC & INS & STA \\
STATE STREET CORP & BA & STT \\
SUN LIFE FINANCIAL INC & INS & SLF \\
SUNTRUST BANKS INC & BA & STI \\
T ROWE PRICE GROUP INC & PB & TROW \\
TORONTO DOMINION BANK ONT & BA & TD \\
U B S AG & BA & UBS \\
U S BANCORP DEL & BA & USB \\
UNITEDHEALTH GROUP INC & INS & UNH \\
WACHOVIA CORP 2ND NEW & BA & WB \\
WADDELL \& REED FINANCIAL INC & PB & WDR \\
WELLS FARGO \& CO NEW & BA & WFC \\
X L CAPITAL LTD & INS & XL \\
\caption{U.S. firm names, sectors, and ticker symbols.  BA: bank, PB: broker-dealer, INS: insurance.}
\label{table: QGC ticker symbol table}
\end{longtable}

\section{List of Indian Banks}
\begin{longtable}[H]{c | c}
\textbf{Firm Name} & \textbf{Ticker Symbol} \\ \hline\hline
ALLAHABAD BANK & ALBK \\
ANDHRA BANK & ANDB \\                       
AXIS BANK & AXSB \\                         
BANK OF BARODA & BOB \\                     
BANK OF INDIA & BOI \\                       
BANK OF MAHARASHTRA & BOMH \\               
BANK OF TRAVANCORE & SBT \\   
CANARA BANK & CBK \\                        
CENTRAL BANK OF INDIA & CBOI \\ 
CITY UNION BANK & CUBK \\                   
CORPORATION BANK & CRPBK \\
DENA BANK & DBNK \\                         
FEDERAL BANK & FB \\
HDFC BANK & HDFCB \\                        
I.N.G. VYSYA BANK & VYSB \\
ICICI BANK & ICICIBC \\                     
IDBI BANK & IDBI \\                         
INDIAN BANK & INBK \\                       
INDIAN OVERSEAS BANK & IOB \\               
INDUSIND BANK & IIB \\                      
JAMMU \& KASHMIR BANK & JKBK \\             
KARUR VYSYA BANK & KVB \\                   
KOTAK MAHINDRA BANK & KMB \\                
ORIENTAL BK OF COMMERCE & OBC \\            
PUNJAB NATIONAL BANK & PNB \\
SOUTH INDIAN BANK & SIB \\                  
STATE BANK OF INDIA & SBIN \\               
STATE BK OF BIN \& JAIPUR & SBBJ \\
SYNDICATE BANK & SNDB \\                    
UCO BANK & UCO \\                           
UNION BANK OF INDIA & UNBK \\               
VIJAYA BANK & VJYBK \\                      
YES BANK & YES \\                          
\caption{India firm names and ticker symbols.}
\label{table: QGC ticker symbol table for India}
\end{longtable}

\end{document}